\documentclass[reqno,12pt]{amsart}
\usepackage[backref]{hyperref}

\usepackage{t1enc}
\usepackage[latin1]{inputenc}
\usepackage[english]{babel}

\usepackage{yfonts}

\usepackage{bbm}
\usepackage{bm}
\usepackage{mathrsfs}

\newcommand{\be}[0]{\begin{equation}}
\newcommand{\ee}[0]{\end{equation}}

\usepackage{amssymb}
\usepackage{amsmath}
\usepackage{amsthm}
\usepackage{mathrsfs}
\usepackage[arrow,matrix,curve]{xy}
\usepackage{graphicx}
\usepackage{amscd,verbatim}
\usepackage{mathtools}
\usepackage{comment}

\usepackage{color}
\definecolor{darkgreen}{cmyk}{1,0,1,.2}
\definecolor{m}{rgb}{1,0.1,1}
\definecolor{green}{cmyk}{1,0,1,0}

\definecolor{test}{rgb}{1,0,0}
\definecolor{cmyk}{cmyk}{0,1,1,0}

\newcommand\dirac{\slash\!\!\!\partial}

\usepackage{fullpage,enumerate}
\newcommand{\vectnu}{{\boldsymbol{\nu}}}
\newcommand{\bea}{\begin{eqnarray}}
\newcommand{\eea}{\end{eqnarray}}
\newcommand{\beq}{\begin{equation}}
\newcommand{\eeq}{\end{equation}}

\newcommand\A{{\bf A}}
\newcommand\B{{\bf B}}
\newcommand\G{\Gamma_{g,\vectnu}}
\newcommand\wt{\widetilde}

\newtheorem{theorem}{Theorem}[section]

\newtheorem{lemma}[theorem]{Lemma}
\newtheorem{proposition}[theorem]{Proposition}
\newtheorem{corollary}[theorem]{Corollary}

\theoremstyle{remark}
\newtheorem{definition}[theorem]{Definition}
\newtheorem{remark}[theorem]{Remark}

\numberwithin{equation}{section}
\numberwithin{theorem}{section}

\newcommand{\hyp}{{\mathbb H}}
\newcommand{\slr}{{\mathbf {PSL}}(2,\R)}
\newcommand{\cl}{{\mathcal L}}

\newcommand\cA{\mathcal{A}}

\newcommand\cE{\mathcal{E}}
\newcommand\cF{\mathcal{F}}

\newcommand\cK{\mathcal{K}}
\newcommand\cM{\mathcal{M}}

\newcommand\cL{\mathcal{L}}

\newcommand\cP{\mathcal{P}}

\newcommand\cT{\mathcal{T}}
\newcommand\cV{\mathcal{V}}

\newcommand\sP{\mathscr{P}}

\newcommand{\Z}{\ensuremath{\mathbb Z}}
\newcommand{\R}{\ensuremath{\mathbb R}}

\newcommand{\CC}{{\mathbb C}}

\newcommand{\NN}{{\mathbb N}}
\newcommand{\QQ}{{\mathbb Q}}
\newcommand{\RR}{{\mathbb R}}
\newcommand{\TT}{{\mathbb T}}
\newcommand{\ZZ}{{\mathbb Z}}

\newcommand{\HH}{{\mathbb H}}
\begin{document}

\title
{Topological phases on the hyperbolic plane: fractional bulk-boundary correspondence}

\author[V Mathai]{Varghese Mathai}

\address[Varghese Mathai]{
Department of Pure Mathematics,
School of  Mathematical Sciences, 
University of Adelaide, 
Adelaide, SA 5005, 
Australia}

\email{mathai.varghese@adelaide.edu.au}

\author[G.C.Thiang]{Guo Chuan Thiang}

\address[Guo Chuan Thiang]{
Department of Pure Mathematics,
School of  Mathematical Sciences, 
University of Adelaide, 
Adelaide, SA 5005, 
Australia}

\email{guochuan.thiang@adelaide.edu.au}

\begin{abstract}
We study topological phases in the hyperbolic plane using noncommutative geometry and T-duality, and show that fractional versions of the quantised indices for integer, spin and anomalous quantum Hall effects can result. Generalising models used in the Euclidean setting, a model for the bulk-boundary correspondence of fractional indices is proposed, guided by the geometry of hyperbolic boundaries.
\end{abstract}

\maketitle


\section{Introduction}
Most of the existing work on topological phases in condensed matter physics deals implicitly with a Euclidean geometry. This is the case both in the well-studied integer quantum Hall effect (IQHE) and in the more recently discovered phases such as the Chern insulator \cite{CZ}, time-reversal invariant topological insulators \cite{Hsieh,Konig}, topological superconductors, and crystalline topological phases. This is partly due to the availability of the classical Bloch--Floquet theory, based upon Fourier transforming the \emph{abelian} (sub)group $\ZZ^d$ of discrete translational symmetries of a crystal lattice in Euclidean space. This leads to the topological band theory paradigm, where vector bundles over the Brillouin torus are constructed out of spectral projections of $\ZZ^d$-invariant Hamiltonians, which then determine topological invariants such as Chern classes, Kane--Mele invariants \cite{KM}, $K$-theory classes \cite{Kitaev,FM,Thiang} etc. 

Besides the problem of cataloguing \emph{bulk} topological phases, there is also the important issue of modelling mathematically the fundamental physical concept of the \emph{bulk-boundary correspondence}. This roughly says that in passing from the bulk of a material hosting a topologically non-trivial gapped phase to the external ``topologically trivial vacuum'', the gapped condition is violated at the boundary (``zero modes'' fill the gap) and the change in bulk topological indices is furthermore recorded in the form of a boundary topological invariant. For the IQHE, this correspondence is the equality of bulk Hall conductance and (direct) edge conductance, which was proved rigorously in the noncommutative geometry (NCG) framework in \cite{Kellendonk1}. Indeed, the non-triviality of bulk topological phases is often deduced experimentally from the detection of boundary gapless modes.

In this paper, we study topological phases in the \emph{hyperbolic plane}, propose a bulk-boundary correspondence of the resulting topological indices, which may be \emph{fractional}, and show its persistence in the presence of certain types of disorder. We also show that with time-reversal symmetry, there is an interesting $\ZZ_2$ ``hyperbolic plane topological insulator'', characterised by a hyperbolic version of the Kane--Mele invariant \cite{KM} originally introduced for the (Euclidean plane) quantum spin Hall effect. 

Our main motivation comes from the possibility of realising fractional indices in the hyperbolic setting, as a result of certain orbifold Euler characteristics taking \emph{rational} rather than integer values, as explained by Marcolli--Mathai in \cite{Marcolli01,Marcolli06}. Physically, the non-Euclidean geometry is supposed to provide an effective model for electron-electron interactions while formally staying within the single-particle framework. Besides providing a predictive model for the fractional quantum Hall effect that can be compared with experiments \cite{Marcolli06}, a recent work of Marcolli--Seipp \cite{Marcolli17} showed that one can even obtain interesting composite fermions and anyon representations by considering symmetric products of Riemann surfaces. The extension of the bulk-boundary correspondence principle from integer to fractional indices, and from Euclidean to other geometries, is therefore of great interest.

We will use tools from NCG, which were first deployed by Bellissard et al to analyse the IQHE \cite{Bellissard}. In that setting, we recall that a mild form of noncommutativity occurred in the sense that quantum Hall Hamiltonians enjoyed only \emph{magnetic translational symmetry}, so that a Brillouin torus is unavailable in the strict sense. An important insight is that the Kubo formula for Hall conductivity, obtainable in the commutative case as an integral of the Chern class of the valence bundle over the Brillouin torus, could be understood as a pairing of a \emph{geometrically} determined cocycle with the $K$-theory class of the Fermi projection. The effect of disorder in producing plateaux could also be accounted for in a rigorous way.

Fortuitously, the NCG framework is very general and allows us to go further and dispense with the Euclidean space paradigm altogether. In hyperbolic geometry, lattice translation symmetry is no longer an abelian group $\ZZ^2$, but some noncommutative surface group $\Gamma_g$, so the classical Bloch--Floquet theory and momentum space are unavailable even in the absence of a magnetic field. Nevertheless, there is a standard NCG prescription of taking a group algebra or crossed product algebra as the noncommutative ``momentum space'', with respect to which bulk topological phases can be defined.

A more serious difficulty arises when trying to formulate a bulk-boundary correspondence in the hyperbolic plane. Our solution is novel: to exploit the idea of \emph{T-duality} to circumvent this difficulty. In the Euclidean case, bulk-to-boundary maps are $K$-theoretic index maps associated to a Toeplitz-like extension \cite{Pimsner} of the bulk algebra of observables with the boundary algebra, where the extended algebra encodes some type of half-space boundary conditions \cite{Kellendonk1}. Abstractly, such extensions may be studied using $K$-homology or Kasparov theory --- indeed a $KK$-theoretic formulation was worked out in \cite{Bourne}. In \cite{MT1,MT2}, we showed that under appropriate T-duality transformations (a geometric Fourier transform closely related to the Baum--Connes assembly map), the Euclidean space bulk-boundary maps simplify into geometric restriction-to-boundary maps, and these results were extended to Nil and Solv geometries in \cite{HMT1,HMT2}. 

In these cases, there is a translational symmetry transverse to the boundary which, together with the longitudinal symmetries of the boundary, recover the symmetries of the bulk. Then our earlier results show that the Toeplitz-like extension correctly encodes the geometric bulk-boundary relationship T-dually. In the hyperbolic plane, transversality properties of ``hyperbolic translations'' are more complicated (Fig.\ \ref{fig:intersection}), and half-space boundary conditions for tight-binding models become very difficult to give explicitly, e.g.\ the atomic sites are no longer simply labelled by a set of integers. Then it becomes \emph{essential} to utilise T-duality to analyse on the geometric side (cf.\ motivation for the Baum--Connes conjecture). Another technical difference in the hyperbolic setting is that the Pimsner--Voiculescu (PV) exact sequence, relevant for computing the $K$-theory of crossed products with $\ZZ^d$, is not available for computing the twisted crossed product with a surface group. Instead, we use the Kasparov spectral sequence that generalises the PV sequence, presenting the computations in the Appendix.

{\bf Outline and main results.}
We first review the quantum Hall effect on the hyperbolic plane and establish some notation and facts about surface and Fuchsian groups in Section \ref{section:IQHE}. In Section \ref{sect:FM}, we introduce noncommutative T-duality for Riemann surfaces, and compute its effects on $K$-theory generators. In Section \ref{section:extensions}, we recall the relation between $C^*$-algebra extensions and $KK$-theory, and introduce the important geometrical notion of a 1-dimensional boundary separating the hyperbolic plane into a bulk and a vacuum region (Fig.\ \ref{fig:hypercycle}-\ref{fig:intersection}). 

Armed with the above tools, we are able to carry out our main objective --- to obtain a bulk-boundary correspondence of fractional indices --- by using the universal coefficients theorem to design a suitable extension that induces a bulk-boundary map. This construction is relevant both for the empirically verified fractional quantum Hall effect, as well as our proposed fractional version of Chern insulators/anomalous Hall effect and Kane--Mele invariants in Section \ref{section:hyperbolicKM}. A central result of this paper is that this bulk-boundary map is T-dual to the geometric restriction-to-boundary map (Theorem \ref{BBCsimplify}).
Pairings with cyclic cohomology are also analysed in Section \ref{section:cyclic}, leading to a fractional bulk-boundary correspondence. In Section \ref{section:generalisations}, we extend these constructions and results to include the effect of disorder (Proposition \ref{BBCsimplifydisordered}).

\tableofcontents

\section{Overview of the quantum Hall effect on the hyperbolic plane}\label{section:IQHE}
The hyperbolic plane analogue of the quantum Hall effect was initially studied in \cite{Comtet, CHMM, CHM}. Quantisation of the Hall conductance followed from similar arguments to those in the IQHE \cite{Bellissard}, with the added bonus that \emph{fractional} indices could be achieved. We begin by reviewing the construction of
magnetic Hamiltonians in a continuous
model with a background hyperbolic geometry term \cite{CHMM}.

One model of
   the hyperbolic plane is the upper half-plane $\hyp$ 
equipped with its usual Poincar\'e metric $(dx^2+dy^2)/y^2$, and
symplectic area form $\omega_\hyp = dx\wedge dy/y^2$. The group $\slr$ acts
transitively and isometrically on $\hyp$ by M\"obius transformations 
$$ 
x+iy = z \mapsto
gz = \frac{az+b}{cz+d},\quad\mbox{for } g=\left(
\begin{array}{cc} a & b\\ c & d
\end{array}\right).
$$ 
Any smooth Riemann surface $\Sigma_g$ of genus $g$ greater than 1 can
be realised as the quotient of
$\hyp$ by the action of its fundamental group realised as a cocompact 
torsion-free discrete subgroup $\Gamma=\Gamma_g$ of $ \slr$.

Pick a 1-form $\A$ such that $d\A = \theta\omega_\hyp$, for some fixed
$\theta \in \R$. As in geometric quantisation we may regard $\A$ as
defining a connection $\nabla = d-i\A$ on a line bundle $\cl$ over
$\hyp$, whose curvature is $\theta\omega_\hyp$. Physically we can think of
$\A$ as the electromagnetic vector potential for a uniform magnetic field
of strength $\theta$ normal to $\hyp$. Using the Riemannian metric the
Hamiltonian of an electron in this field is given in suitable units by $$H
= H_\A = \frac 12\nabla^*\nabla = \frac 12(d-i\A)^*(d-i\A).$$ 
Comtet \cite{Comtet} has computed
the spectrum of the unperturbed Hamiltonian $H_\A$, for $\A = -\theta
dx/y$, to be the union of finitely many eigenvalues $\{(2k+1)\theta
-k(k+1):k=0,1,2\ldots < |\theta|-\frac 12\}$, and the continuous spectrum
$[\frac 14 + \theta^2, \infty)$. Any $\A$ is cohomologous to $-\theta
dx/y$ (since they both have $\omega_\hyp$ as differential) and forms
differing by an exact form $d\phi$ give equivalent models: in fact,
multiplying the wave functions by $\exp(i\phi)$ shows that the models for
$\A$ and $-\theta dx/y$ are unitarily equivalent. This equivalence also
intertwines the $\Gamma$-actions so that the spectral densities for the two
models also coincide.

This Hamiltonian can
be perturbed by adding a potential term $V$.
In \cite{CHMM}, the authors took
$V$ to be invariant under $\Gamma$, while in \cite{CHM}, the authors
allowed any smooth random potential function $V$ on $\mathbb H$
using two general notions of random potential
(in the literature random usually refers to
the $\Gamma$-action on the disorder space
being required to admit an ergodic
invariant measure). 

However, the perturbed Hamiltonian $H_{\A,V} = H_\A +V$, which
is important for the quantum Hall effect, has unknown spectrum
for general $\Gamma$-invariant $V$. For $\gamma \in \Gamma$,
let $\psi_\gamma(x) $ be a function on $\hyp$
satisfying $\gamma^*\A-\A = d \psi_\gamma$, such that
$\psi_\gamma(x_0)=0$ for all $\gamma  \in \Gamma$.

Define a projective unitary action $T^\sigma$ of $\Gamma$ on $L^2(\hyp)$ as follows:
\begin{align*}
U_\gamma(f)(x) & = f(\gamma^{-1}x)\\
S_\gamma(f)(x) & = \exp(-2\pi i \psi_\gamma(x)) f(x)\\
T_\gamma^\sigma &=U_\gamma\circ S_\gamma.
\end{align*}
Then the operators $T_\gamma^\sigma$, also known as {\em magnetic translations}, satisfy $T^\sigma_e={\rm{Id}}, \,\, T^\sigma_{\gamma_1}
T^\sigma_{\gamma_2} = \sigma(\gamma_1, \gamma_2)T^\sigma_{\gamma_1\gamma_2}$,
where 
\beq
\sigma(\gamma, \gamma') = \exp(-2\pi i\theta \psi_\gamma(\gamma') ), \label{multiplier}
\eeq
which is a {\em multiplier} (or {\em 2-cocycle})
on $\Gamma$, that is, it satisfies,
\begin{enumerate}
\item $\sigma(\gamma, e)=\sigma(e, \gamma)=1$ for all $\gamma \in \Gamma$;
\item $\sigma(\gamma_1, \gamma_2)\sigma(\gamma_1\gamma_2, \gamma_3)=\sigma(\gamma_1, \gamma_2\gamma_3)\sigma(\gamma_2, \gamma_3)$
for all $\gamma_j \in \Gamma,\, j=1,2,3.$
\end{enumerate}
These are the multipliers that we will be interested in in this paper, so that $\sigma=\sigma_\theta$ will be understood to be parametrised by $\theta$. An easy calculation shows that $T^\sigma_\gamma \nabla = \nabla T^\sigma_\gamma$ and taking adjoints, 
$T^\sigma_\gamma \nabla^* = \nabla^* T^\sigma_\gamma$. Therefore $T^\sigma_\gamma H_\A = H_\A T^\sigma_\gamma$.
Also, since $V$ is $\Gamma$-invariant, $T^\sigma_\gamma V = V T^\sigma_\gamma$. We conclude that for all $\gamma\in \Gamma$,
$T^\sigma_\gamma H_{\A,V} = H_{\A, V} T^\sigma_\gamma$, that is, the Hamiltonian commutes with magnetic translations.
The commutant of the projective action $T^\sigma$ is the projective action $T^{\bar\sigma}$. If $\lambda$ lies in a spectral gap
of $H_{\A, V} $, then the Riesz projection is $p_\lambda(H_{\A, V})$ where $p_\lambda$ is a  smooth compactly supported
function which is identically equal to 1 in the interval $[{\rm inf\, spec} H_{\A, V}, \lambda]$, 
and whose support is contained in the interval $[-\varepsilon +{\rm inf\, spec} H_{\A, V}, \lambda+ \varepsilon]$ for some $\varepsilon>0$ small enough.
Then  
\beq
p_\lambda(H_{\A, V})\in C^*_r(\Gamma, \bar\sigma)\otimes\cK(L^2(\cF)), \label{spectralprojection}
\eeq
where $\cF$ is a fundamental
domain for the action of $\Gamma$ on $\hyp$, and $p_\lambda(H_{\A, V})$ defines an element in 
$K_0( C^*_r(\Gamma, \bar\sigma))$. Here $C^*_r(\Gamma, \bar\sigma)$ is the twisted reduced $C^*$-algebra of $\Gamma$. By the {\em gap hypothesis}, the Fermi level of the physical system modeled by the Hamiltonian $H_{\A, V}$ lies in a spectral gap.

{\bf Fuchsian groups.} As in \cite{Marcolli99,Marcolli01}, we can even take $\Gamma$ to be a Fuchsian group of signature $(g,\vectnu)\equiv(g,\nu_1,\ldots,\nu_r)$, that is, $\Gamma=\Gamma_{g,\vectnu}$ is a discrete cocompact subgroup of $\slr$ of genus $g\geq 0$ and with $r$ elliptic elements of order $\nu_1,\ldots,\nu_r$ respectively. The canonical presentation is
\beq   
\begin{split}
    \Gamma_{g,\vectnu}  = &\langle A_i,B_i,C_j, i = 1,\ldots,g, j = 1,\ldots,r,\,|\\
    & [A_1,B_1]\ldots[A_g,B_g]C_1\ldots C_r = 1,\, C_j^{\nu_j} = 1 \rangle.
    \end{split}
\eeq
and the quotient orbifold $\Sigma_{g,\vectnu}=\HH/\Gamma_{g,\vectnu}$ is a compact oriented surface of genus $g$ with $r$ elliptic points $p_1,\ldots,p_r$. Note that $\HH/\Gamma_g = \Sigma_g, g\geq 2$ considered above is just the special case where $r=0$. Each $p_j$ is a conical singularity in the sense that it looks locally like a quotient $D^2/\ZZ_{\nu_j}$ under rotation by $2\pi/\nu_j$, where $D^2$ is a unit disc in $\RR^2$. The universal orbifold covering space of $\Sigma_{g,\vectnu}$ is $\HH$, and the \emph{orbifold fundamental group} of $\Sigma_{g,\vectnu}$ \cite{Scott} recovers $\Gamma_{g,\vectnu}$. Each orbifold $\Sigma_{g,\vectnu}$ is ``good'' in the sense that there is a (non-unique) smooth $\Sigma_{g'}$ which covers $\Sigma_{g,\vectnu}$, with projection map the quotient under the action of a finite group $G$ on $\Sigma_{g'}$, and the Riemann--Hurwitz formula gives
$$
g'=1+\frac{\#G}{2}(2(g-1)+(r-\nu))
$$
where $\nu\coloneqq\sum_{j=1}^r \frac{1}{\nu_j}$. There is a short exact sequence
\beq
1\longrightarrow \Gamma_{g'}\longrightarrow\Gamma_{g,\vectnu}\longrightarrow G\longrightarrow 1,\label{SES}
\eeq
so that $\Gamma_{g,\vectnu}$ always contains hyperbolic elements (even if $g<2$).

As before, we will consider multipliers $\sigma$ on $\Gamma_{g,\vectnu}$ defined through the magnetic field by Eq.\ \eqref{multiplier}, which has vanishing Dixmier--Douady invariant, $\delta(\sigma)=0$. The above discussion leading to
Eq.\ \eqref{spectralprojection} is still valid \cite{Marcolli99}, and we will therefore need the computations \cite{Farsi,Marcolli99}
\beq
K_\bullet(C^*_r(\Gamma_{g,\vectnu},\sigma))\cong \begin{cases} \ZZ^{2+\sum_{j=1}^r (\nu_j-1)}\;\qquad \bullet=0, \\ \ZZ^{2g}\,\qquad\qquad\qquad\;\; \bullet=1. \end{cases}\label{fuchsianKtheory}
\eeq 
We remark that $\Gamma_{g,\vectnu}$ is $K$-amenable \cite{Cuntz}, so that Eq.\ \eqref{fuchsianKtheory} holds also for the full (unreduced) twisted group algebra $C^*(\Gamma_{g,\vectnu},\sigma)$.

{\bf Notation:} Generally, we will use $\Gamma$ (resp.\ $\Sigma$) to denote $\Gamma_{g,\vectnu}$ (resp. $\Sigma_{g,\vectnu}$), and include subscripts only when we need to distinguish the torsion free situation $\Gamma_g, \Sigma_g$ (with $\vectnu$ empty) from the general one.

\section{Noncommutative T-duality for Riemann surfaces}
\label{sect:FM}
T-duality describes an inverse mirror relationship between a pair of type II string theories. Mathematically, it is a geometric analogue of the 
Fourier transform \cite{H}, giving rise to a bijection of (Ramond--Ramond) fields and their $K$-theoretic charges. The global aspects of topological T-duality are most interesting in the presence of a background flux \cite{BEM}, while the noncommutative generalisation appears in \cite{MR05,MR06,MR06-2}. This body of work pertains to (noncommutative) circle or torus bundles. In this section, we introduce the notion of (noncommutative) T-duality for Riemann surfaces. This is motivated by the fact that the twisted group algebra $C^*_r(\Gamma_g, \sigma)$ of the surface group $\Gamma_g$ appears as the bulk algebra when studying the IQHE on the hyperbolic plane.

As a warm up, we recall the notion of T-duality $T_{\rm circle}$ for circles. It can be defined as the composition of Poincar\'{e} duality $K^0(S^1)\cong K_1(S^1)$ with the Baum--Connes isomorphism\footnote{The Pontryagin dual $\TT$ of $\ZZ$ is also a circle, but we use a different symbol from $S^1=B\ZZ$ for clarity.} $\mu^\ZZ:K_1(S^1)\equiv K_1(B\ZZ)\cong K_1(C^*(\ZZ))\cong K^{-1}(\TT)$, and its formula on generators is
\begin{align}
K^0(S^1) \ni [{\bf 1}_{S^1}] &\overset{T_{\rm circle}} {\longleftrightarrow}[\zeta]\in K^{-1}(\TT)\nonumber \\
K^{-1}(S^1)\ni [W] & \overset{T_{\rm circle}} {\longleftrightarrow}[{\bf 1}_\TT]\in K^0(\TT).\label{circleTduality}
\end{align}
Here, ${\bf 1}_{S^1}, {\bf 1}_\TT$ are trivial line bundles generating $K^0(S^1), K^0(\TT)$ respectively, while $W$ and $\zeta$ are winding number 1 unitaries in $C(S^1)$ and $C(\TT)$ generating $K^{-1}(S^1)$ and $K^{-1}(\TT)$ respectively. This deceptively simple formula hides the non-triviality of the Baum--Connes map --- in particular, that it is an isomorphism (see, e.g.\ the detailed discussion in Section 4 of \cite{Valette}).

There is also a formulation as a Fourier--Mukai transform \cite{BEM} which makes the analogy to the Fourier transform more apparent, and proceeds via a ``push-pull'' construction. In the following, we generalise the latter construction for $\Sigma$. This proceeds quite abstractly, and the reader may skip to Section \ref{section:NCdualityeven} for the description through the Baum--Connes map.

For $\Sigma=\Sigma_{g,\vectnu}$ an orbifold, the appropriate algebra of functions is $C(\Sigma_{g'})\rtimes G$ with the $G$-action in accordance with Eq.\ \eqref{SES}. For the rest of this subsection, we will abuse notation and simply write $C(\Sigma)$ to mean $C(\Sigma_{g'})\rtimes G$, keeping in mind that we are treating $\Sigma$ as an orbifold rather than just an ordinary topological space, thus it has orbifold $K$-theories etc. \cite{ALR}. We note that $G$ is finite so that there are Green--Julg--Rosenberg isomorphisms $K_\bullet(C(\Sigma))\equiv K^\bullet_{\rm orb}(\Sigma)\equiv K^\bullet_G(\Sigma_{g'})\cong K_\bullet(C(\Sigma_{g'})\rtimes G)$ and $K^\bullet(C(\Sigma))\equiv K_\bullet^{\rm orb}(\Sigma)\equiv K_\bullet^G(\Sigma_{g'})\cong K^\bullet(C(\Sigma_{g'})\rtimes G)$, cf.\ Theorem 20.2.7 of \cite{Blackadar}.

Consider the diagram,
\begin{equation}
\xymatrix{ 
& {\cM_\sigma} \ar[d] & \\
 &  C(\Sigma)\otimes C^*_r(\Gamma, \sigma)  &  \\
 C(\Sigma) \ar[ur]_{\iota_1} && C^*_r(\Gamma, \sigma)   \ar[ul]^{\iota_2}, 
 }\nonumber
\end{equation}
where 
$\iota_1, \iota_2$ are inclusion maps, and
$\cM_\sigma$ is the universal finite projective module over $C(\Sigma)\otimes C^*_r(\Gamma, \sigma)$, playing the role of the Poincar\'e line bundle, which we now construct.

Consider  $\Gamma_g$ as a discrete subgroup of $\slr$, acting on the right on $\hyp$.
The multiplier $\sigma=\sigma_\theta$ extends to $\slr$ \cite{CHMM} as follows. Let $c$ denote the area cocycle on $\slr$ as defined just after Proposition \ref{BC}, then $\sigma=\sigma_\theta = \exp(2\pi \sqrt{-1} \theta c)$. So there is a central extension,
$$
1\longrightarrow \mathrm{U}(1)\longrightarrow \slr^\sigma \longrightarrow\slr\longrightarrow 1.
$$
The cocycle $\sigma$ has the property that it vanishes on ${\rm SO}(2)$, the stabiliser of $\slr$ on $\HH$, so we may define $\hyp^\sigma={\rm SO(2)}\backslash\slr^\sigma$. There are also the central extensions restricted to $\Gamma=\Gamma_{g,\vectnu}$ and to its torsion-free normal subgroup $\Gamma_{g'}$
\begin{align*}
&1\longrightarrow \mathrm{U}(1)\longrightarrow \Gamma_{g'}^\sigma \longrightarrow\Gamma_{g'}\longrightarrow 1,\\
&1\longrightarrow \mathrm{U}(1)\longrightarrow \Gamma_{g,\vectnu}^\sigma \longrightarrow\Gamma_{g,\vectnu}\longrightarrow 1
\end{align*}
Now $\Gamma_{g'}^\sigma$ acts on the right on ${\rm SO(2)}\backslash\slr^\sigma = \hyp^\sigma$
and by the left action on $C^*_r(\Gamma, \sigma)$. Set $\cM_\sigma = C(\Sigma_{g'}, \cV_\sigma)$, where 
\beq
\cV_\sigma =  \hyp^\sigma\times_{\Gamma_{g'}^\sigma} C^*_r(\Gamma, \sigma) \nonumber
\eeq 
is a locally trivial vector bundle with fibers $C^*_r(\Gamma, \sigma)$ over
$\hyp^\sigma/\Gamma_{g'}^\sigma = \Sigma_{g'}$. There remains a left action of the quotient $G= \Gamma_{g,\vectnu}/\Gamma_{g'}\cong\Gamma_{g,\vectnu}^\sigma/\Gamma_{g'}^\sigma$ on $\cV_\sigma$ covering that on $\Sigma_{g'}$.

Recall that $\Sigma_g=\HH/\Gamma_g=E\Gamma_g/\Gamma_g$ is a classifying space $B\Gamma_g$. It is known that $C^*_r(\Gamma_g, \sigma)$ is $K$-oriented with Poincar\'{e} duality $K_0(C^*_r(\Gamma_g, \sigma)) \cong K^0(C^*_r(\Gamma_g, \bar{\sigma}) )$
implemented by a fundamental class constructed in Section 2.6 of \cite{BMRS} (see also Theorem 3.3 of \cite{EEK}), and that $\Sigma_g$ is $K$-oriented. Their $K$-theories are related by twisted versions of Kasparov's maps (pp.\ 192 of \cite{Kas88}),
$$ \beta_\sigma:K_\bullet(\Sigma_g)\rightarrow K_\bullet(C^*(\Gamma_g, \sigma)),$$
$$\alpha_{\bar{\sigma}}: K^\bullet(C^*(\Gamma_g,\bar{\sigma}))\rightarrow K^\bullet(\Sigma_g),$$
which turn out to be dual to each other and isomorphisms in this case, as a result of K-orientability. For $\Sigma=\Sigma_{g,\vectnu}$, there is Poincar\'{e} duality in the orbifold sense.
More precisely, by section 4, \cite{Kas88}, one has $G$-equivariant Poincar\'e duality, which implies that 
$$
PD : K^G_\bullet(\Sigma_{g'}) \cong K_G^\bullet(\Sigma_{g'}),
$$
that is,
$$
PD_{orb} : K^{\rm orb}_\bullet(\Sigma) \cong K_{\rm orb}^\bullet(\Sigma).
$$
We say that $\Sigma$ is $K_{\rm orb}$-oriented in this case.
 We use the $G$-equivariant version of Kasparov's maps $\alpha_{\bar{\sigma}}, \beta_\sigma$ to deduce $K$-orientability for $C^*(\G,\sigma)$ (and $C^*_r(\G,\sigma)$ by $K$-amenability). Thus for $\Sigma=\Sigma_{g,\vectnu}$ and $\Gamma=\G$, we have that $C(\Sigma)$ and $C^*_r(\Gamma,\sigma)$ are $K$-orientable, and together with the fact that their $K$-theories are torsion-free, we deduce by the K\"unneth theorem that $C(\Sigma)\otimes C^*_r(\Gamma, \sigma)  $ is $K$-orientable as well. Thus we have the Poincar\'{e} dualities
\beq
PD_{C(\Sigma)\otimes C^*_r(\Gamma, \sigma) } : K_0(C(\Sigma)\otimes C^*_r(\Gamma, \sigma) ) \cong K^0(C(\Sigma)\otimes  C^*_r(\Gamma, \bar{\sigma}) ), \nonumber
\eeq
\beq
PD_{C^*_r(\Gamma, \sigma) } : K^0(C^*_r(\Gamma, \bar{\sigma}) ) \cong K_0(C^*_r(\Gamma, \sigma) ). \nonumber
\eeq

Finally, for an orbifold vector bundle $\cE$ over $\Sigma=\Sigma_{g,\vectnu}$ (or $G$-equivariant bundle over $\Sigma_{g'}$) representing a class in $K_0(C(\Sigma))$, noncommutative T-duality is the composition,
\beq
K^0(C(\Sigma))\ni [\cE] \longrightarrow 
\iota_2^!( (\iota_1)_*([\cE]) \otimes_{C(\Sigma)}  \cM_\sigma) \in K_0(C^*_r(\Gamma, \sigma) ), \nonumber
\eeq
where the wrong way map, or Gysin map $$\iota_2^!\colon K_0(C(\Sigma)\otimes C^*_r(\Gamma, \sigma) ) \to K_0(C^*_r(\Gamma, \sigma) ) $$
is defined by 
$$\iota_2^! = PD_{C^*_r(\Gamma, \sigma) }\circ (\iota_2)^*\circ PD_{C(\Sigma)\otimes C^*_r(\Gamma, \sigma) },$$ where 
$$ (\iota_2)^*\colon K^0(C(\Sigma)\otimes C^*_r(\Gamma, \bar{\sigma}) ) \to K^0(C^*_r(\Gamma, \bar{\sigma}) ) $$ is the homomorphism in $K$-homology.

\subsection{Noncommutative T-duality for $\Sigma_g$ in even degree}\label{section:NCdualityeven}
We can also formulate noncommutative T-duality through the Baum--Connes isomorphisms. First, consider the case where $\vectnu$ is empty, and the $K$-theory degree is even.
\subsubsection*{Poincar\'e duality in $K$-theory}
The 2D Riemann surface $\Sigma_g$ is a Spin manifold, therefore it is $K$-oriented. Poincar\'e duality in the $K$-theory of $\Sigma_g$ is given by 
\begin{align}
PD_{\Sigma_g}: K^0(\Sigma_g) &\stackrel{\sim}{\longrightarrow} K_0(\Sigma_g)\nonumber\\
[\mathcal{E}] & \mapsto [\dirac_{\Sigma_g}\otimes \mathcal{E}] \label{PDformula}
\end{align}
where  $\dirac_{\Sigma_g} \otimes \mathcal{E}$ is the Spin Dirac operator on $\Sigma_g$ coupled to the vector bundle $\mathcal{E}$ over $\Sigma_g$. In particular, we see that the class of the trivial line bundle $[{\bf 1}_{\Sigma_g}]$ maps to $ [\dirac_{\Sigma_g}]$, the class of the 
Spin Dirac operator on $\Sigma_g$, also known as the fundamental class in $K$-homology, which is a generator. Also the class of 
the nontrivial line bundle $\cL$ with Chern number $c_1(\cL)=\int_{\Sigma_g} \frac{dx \wedge dy}{y^2}=1$ maps under $PD_{\Sigma_g}$ to the class of the coupled Spin Dirac operator,  $ [\dirac_{\Sigma_g} \otimes \cL]$. 
Recall that 
\beq
K^0(\Sigma_g) \cong \ZZ[{\bf 1}_{\Sigma_g}] \oplus \ZZ[\cL],\nonumber
\eeq 
and since Poincar\'e duality is an isomorphism, we see that
\beq
K_0(\Sigma_g) \cong \ZZ[\dirac_{\Sigma_g}] \oplus \ZZ[\dirac_{\Sigma_g}\otimes \cL].\nonumber
\eeq

\subsubsection{Twisted Baum-Connes isomorphism}
Let $[\wt{\cL}]$ denote the class of the virtual bundle $\cL\ominus{\bf 1}_{\Sigma_g}$, then we can take $[{\bf 1}_{\Sigma_g}], [\wt{\cL}]$ as generators for $K^0(\Sigma_g)$.

Recall that $\Sigma_g\simeq B\Gamma_g$, so the twisted Baum--Connes map \cite{CHMM,Mathai99} is an isomorphism of groups, 
\beq
\mu_\theta \colon K_0(\Sigma_g) \stackrel{\sim}{\longrightarrow}  K_0(C^*_r(\Gamma_g, \sigma)).\nonumber
\eeq
because the Baum--Connes conjecture with coefficients is true for $\Gamma_g$ (cf. \cite{BCH,Connes94}). It can be expressed as 
\beq
\mu_\theta:K_0(\Sigma_g) \ni[\dirac_{\Sigma_g} \otimes \mathcal{E}] \mapsto {\rm index}_{C^*_r(\Gamma_g, \sigma)}(\dirac_{\Sigma_g} \otimes \mathcal{E} \otimes \cV_\sigma)
\in K_0(C^*_r(\Gamma_g, \sigma)).\nonumber
\eeq

Let ${\bf 1}$ denote the trivial projection in $C^*_r(\Gamma_g, \sigma)$.
\begin{proposition}\label{BC}
$\mu_\theta$ exchanges $[\dirac_{\Sigma_g}\otimes \wt{\cL}]  \leftrightarrow [\bf 1]$ and $ [\dirac_{\Sigma_g}]  \leftrightarrow [\cP_\sigma]$, where $[\cP_\sigma]$ denotes the nontrivial class in $K_0(C^*_r(\Gamma_g, \sigma))$ defined below.
\end{proposition}
\begin{proof}
Let $\tau$ denote the von Neumann trace on $C^*_r(\Gamma_g, \sigma)$ extended to an additive map {$K_0(C^*_r(\Gamma_g, \sigma))\rightarrow \RR$}, and $\B = \theta\, (dx\wedge dy)/y^2$ a 2-form on $\Sigma_g$. By the index theorem in \cite{CHMM,Mathai99},
\beq
\tau(\mu_\theta([\dirac_{\Sigma_g}\otimes \mathcal{E}])) = \int_{\Sigma_g} e^\B \wedge {\rm Ch}(\mathcal{E}) = {\rm rank}(\mathcal{E})  \int_{\Sigma_g} \B + c(\mathcal{E}) =  {\rm rank}(\mathcal{E})\,  \theta + c(\mathcal{E}). \nonumber
\eeq
Therefore 
\beq
\tau(\mu_\theta([\dirac_{\Sigma_g}\otimes \wt{\cL}]) =1,\qquad
\tau(\mu_\theta([\dirac_{\Sigma_g}\otimes {\bf 1}_{\Sigma_g}])) =\theta\nonumber
\eeq
and the range of $\tau$ is $\ZZ+\theta\ZZ$ with $\tau([{\bf 1}])=1$.
Recall that the {\em area cocycle} $c(g_1, g_2)$ on $\slr$ is a group 2-cocycle defined as the oriented hyperbolic area of the geodesic triangle with vertices at $\{o, g_1.o, g_2^{-1}.o\}$ on the hyperbolic plane $\hyp$ with $o \in \hyp$, as in Fig.\ \ref{fig:hyp-triangle}, and let $\tau_c$ be the corresponding cyclic cocycle (see Section \ref{sec:cyclicHall}). The higher twisted index theorem (section 2, \cite{Marcolli01}) gives
\beq
\tau_c(\mu_\theta([\dirac_{\Sigma_g}\otimes \mathcal{E}]))=\int_{\Sigma_g}\omega_c\wedge e^\B \wedge {\rm Ch}(\mathcal{E}) = {\rm rank}(\mathcal{E}) 2(g-1),\nonumber
\eeq
where $\omega_c$ is the hyperbolic volume form associated to the area cocycle $c$ on $\Gamma_g$. Then the range of $\tau_c$ on $K_0(C^*_r(\Gamma_g,\sigma))$ is $2(g-1)\ZZ$. Since $\tau_c([{\bf 1}])=\tau_c({\bf 1, 1, 1})=0$, there is another generator $[\cP_\sigma]$ of $K_0(C^*_r(\Gamma_g,\sigma))$ which maps to $2(g-1)$ under $\tau_c$. This other generator is only specified up to some multiple of $[{\bf 1}]$, and we choose it\footnote{Note the slight abuse of notation, since $[\cP_\sigma]$ may actually need be written as a difference of projections.} such that $\tau([\cP_\sigma])=\theta$. The two index formulae allow us to conclude that
\beq
\mu_\theta([\dirac_{\Sigma_g}\otimes \wt{\cL}]) = [{\bf 1}], \qquad \mu_\theta([\dirac_{\Sigma_g}\otimes{\bf 1}_{\Sigma_g}]) =[\cP_\sigma].\nonumber
\eeq

\end{proof}
\begin{figure}[h]
    \centering
    \includegraphics[width=0.5\textwidth]{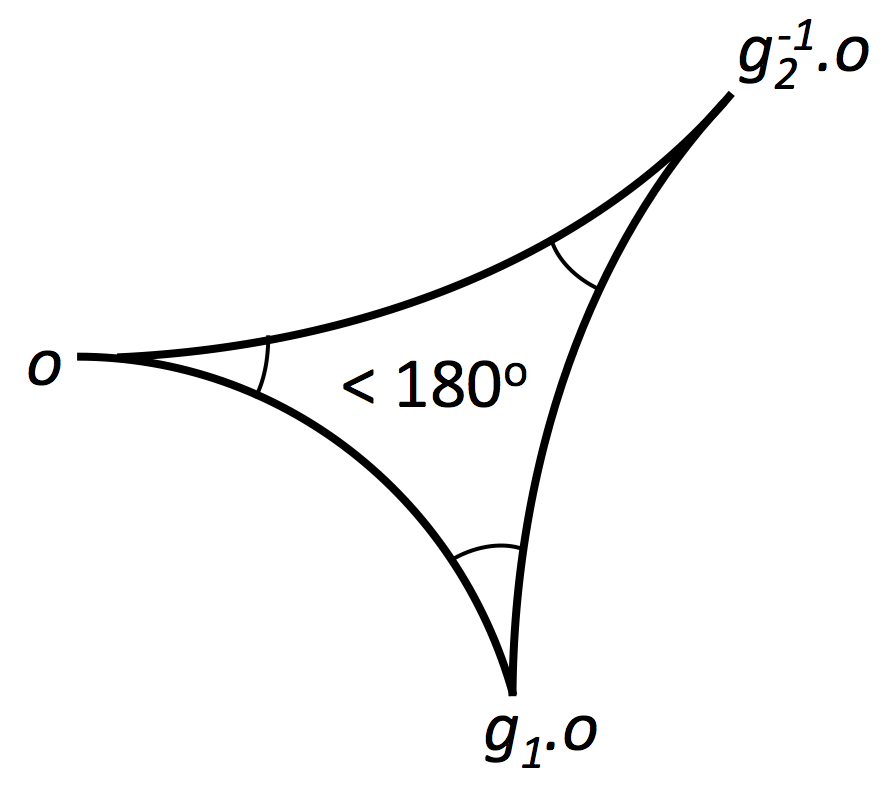}
    \caption{Hyperbolic triangle}
    \label{fig:hyp-triangle}
\end{figure}

Noncommutative T-duality at the level of $K$-theory groups is the composition, 
\beq
T_{\Sigma_g}=\mu_\theta \circ PD_{\Sigma_g} \colon K^0(\Sigma_g) \stackrel{\sim}{\longrightarrow} K_0(C^*_r(\Gamma_g, \sigma)).\nonumber
\eeq
By Eq.\ \eqref{PDformula} and Proposition \ref{BC}, we have

\begin{corollary} \label{cor:evengenformula1}
$T_{\Sigma_g}$ exchanges  $[{\bf 1}_{\Sigma_g}]  \leftrightarrow [\cP_\sigma]$ and $[\wt{\cL}]  \leftrightarrow [\bf 1]$.
\end{corollary}

\begin{remark}\label{rem:cohomologous}
From the physical perspective, $\tau$ is relevant for \emph{gap-labelling} problems, while the geometrically defined 2-cocycle $\tau_c$ turns out to be cohomologous to the (hyperbolic) \emph{Kubo conductivity} cocycle $\tau_K$ \cite{CHMM} which computes the contribution to the Hall conductance by a projection $\cP$. Thus $[\cP_\sigma]$ has the physical meaning of the $K$-theory class contributing to the smallest nonzero value of the quantised Hall conductance.
\end{remark}

\begin{remark}
In the Euclidean case where $\Gamma=\ZZ^2$ and $C^*_r(\ZZ^2,\sigma)$ is the noncommutative torus, $\cP_\sigma$ is the Rieffel projection, cf.\ \cite{MT2}.
\end{remark}


\subsection{Noncommutative T-duality 
for $\Sigma_{g,\nu}$ 
in even degree}\label{sect:fuchsianTduality}

When $\vectnu$ is nonempty, we need to consider the \emph{orbifold} $K$-theory $K^0_{\rm orb}(\Sigma_{g,\vectnu})\cong K^0_G(\Sigma_{g'})$. Note that for the smooth manifold $\Sigma_g$, we have $K^0_{\rm orb}(\Sigma_g)=K^0(\Sigma_g)$. Returning to $\Sigma_{g,\vectnu}$, besides the trivial line bundle ${\bf 1}_{\Sigma_{g,\vectnu}}$ and the (virtual) line bundle $\wt{\cL}$ whose Chern class generates the top degree cohomology, there are new \emph{orbifold line bundles} on $\Sigma_{g,\vectnu}$ (or $G$-equivariant bundles on $\Sigma_{g'}$) which generate extra copes of $\ZZ$ in $K^0_{\rm orb}(\Sigma_{g,\vectnu})$. These extra bundles can be labelled by the non-trivial characters $\chi_j$ of $\ZZ_{\nu_j}$ at each singular point $p_j$ \cite{Marcolli99}. We write $\wt{\cL}_{\chi_j}$ for these virtual bundles, and there are $\sum_{j=1}^r(\nu_j-1)$ classes of them. Then $[{\bf 1}_{\Sigma_{g,\vectnu}]}, [\wt{\cL}]$ and $[\wt{\cL}_{\chi_j}], j=1,\ldots,r$ account for $$K^0_{\rm orb}(\Sigma_{g,\vectnu})\cong \ZZ^{2+\sum_{j=1}^r(\nu_j-1)}.$$


Poincar\'{e} duality $PD_{\Sigma_{g,\vectnu}}$ takes $[\cE]\in K^0_{\rm orb}(\Sigma_{g,\vectnu})$ to the class of the $\cE$-twisted Dirac operator (denoted $\dirac_\cE^+$ in \cite{Marcolli01}) in $K^{\rm orb}_0(\Sigma_{g,\vectnu})\cong K_0^{\Gamma_{g,\vectnu}}(\underline{E}\Gamma_{g,\vectnu})$ with the latter isomorphism given by lifting $\dirac_\cE^+$ to a $\Gamma_{g,\vectnu}$-invariant operator $\wt{\dirac_\cE^+}$ on the contractible cover $\HH\simeq \underline{E}\Gamma_{g,\vectnu}$. Recall also that the twisted Baum--Connes assembly map $\mu_\theta$ gives an isomorphism
\beq
\mu_\theta:K_0^{\Gamma_{g,\vectnu}}(\underline{E}\Gamma_{g,\vectnu})\rightarrow K_0(C^*_r(\Gamma_{g,\vectnu},\sigma)),\nonumber
\eeq
which is in accordance with Eq.\ \eqref{fuchsianKtheory}. Noncommutative T-duality in this case is the map $T_{\Sigma_{g,\vectnu}}=\mu_\theta\circ PD_{\Sigma_{g,\vectnu}}$,
where
\beq
T_{\Sigma_{g,\vectnu}}:K^0_{\rm orb}(\Sigma_{g,\vectnu}) \ni [\mathcal{E}] \mapsto {\rm index}_{C^*_r(\Gamma_{g,\vectnu}, \sigma)}(\dirac_\cE^+ \otimes \cV_\sigma)
\in K_0(C^*_r(\Gamma_{g,\vectnu}, \sigma)).\nonumber
\eeq
There is again a higher index formula \cite{Marcolli99,Marcolli01}
\beq
\tau_c(T_{\Sigma_{g,\vectnu}}([\cE]))=\phi \,{\rm rank}(\cE),\label{higherindexformula}
\eeq
where $\phi=2(g-1)+(r-\nu)\in\QQ$ is the orbifold Euler characteristic of $\Sigma_{g,\vectnu}$. Let $[\cP_\sigma]\in K_0(C^*_r(\Gamma_{g,\vectnu},\sigma))$ be a generator such that $\tau_c(\cP_\sigma)=\phi$. Then we have
\begin{corollary}\label{cor:evengenformula2}
$T_{\Sigma_{g,\vectnu}}$ takes $[\wt{\cL}],[\wt{\cL}_{\chi_j}]$ into ${\rm ker}\, \tau_c$, and $[{\bf 1}_{\Sigma_{g,\vectnu}}]  \mapsto [\cP_\sigma]$ up to an element in ${\rm ker}\, \tau_c$.
\end{corollary}
Note that Corollary \ref{cor:evengenformula2} is consistent with the special case in Corollary \ref{cor:evengenformula1}.

\subsection{Noncommutative T-duality 
for $\Sigma_{g,\nu}$ 
in odd degree}\label{section:NCTdualityodd}

For the T-duality isomorphism $K^{-1}_{\rm orb}(\Sigma_{g,\vectnu})\leftrightarrow K_1(C^*_r(\Gamma_{g,\vectnu}))$, it is convenient to identify both the $K$-groups with $\ZZ^{2g}$ in terms of the canonical $2g$ group generators $A_i,B_i$ of $\Gamma_{g,\vectnu}$. 

First, consider the torsion free case where $\vectnu$ is empty. The abelianisation of $\Gamma_g$ is $\Gamma_g^{\rm ab}\cong\ZZ^{2g}=\ZZ^g\oplus\ZZ^g$ with canonical generators denoted $A_i^{\rm ab}, B_i^{\rm ab}$, and we can also identify $\Gamma_g^{\rm ab}\cong H_1(\Gamma_g,\ZZ)\cong H_1(\Sigma_g)$ with the generating cycles of the latter denoted $l_{A_i}, l_{B_i}$. More generally, $Y^{\rm ab}\in\Gamma_g^{\rm ab}$ has a corresponding homology class $[l_{Y^{\rm ab}}]\in H_1(\Sigma_g)$. 

 Let $L^\sigma_Y$ denote $\sigma$-left translation by $Y\in\Gamma_g$, which is a unitary in $C^*_r(\Gamma_g,\sigma)$. The inclusion $Y\mapsto L^\sigma_Y$ induces a homomorphism 
$\wt{\beta}_a^\sigma:\Gamma_g\rightarrow K_1(C^*_r(\Gamma_g,\sigma))$ which factors through 
$$\beta_a^\sigma:\Gamma_g^{\rm ab}\rightarrow K_1(C^*_r(\Gamma_g,\sigma)).$$
In fact, $\beta_a^\sigma$ is an isomorphism here, so that $[L_{A_i}^\sigma],[L_{B_i}^\sigma]$ are canonical generators for $K_1(C^*_r(\Gamma_g,\sigma))$ \cite{Mathai06} (in the untwisted case the rational injectivity of $\beta_a$ is a general result of \cite{EN87,BV96}). There is also a canonical homomorphism
$$\beta_t: \Gamma_g^{\rm ab}\rightarrow K_1(\Sigma_g)$$
such that 
\beq
\beta_a^\sigma=\mu_\theta\circ\beta_t\label{valetteidentity}
\eeq
where $\mu_\theta=\mu_\theta^{\Gamma_g}$ is the twisted Baum--Connes assembly map $K_1(B\Gamma_g\simeq\Sigma_g)\rightarrow K_1(C^*_r(\Gamma_g,\sigma))$ \cite{Valette,Mathai06}. It is convenient to identify $K_1(\Sigma_g)$ with $H_1(\Gamma_g)$ under the Chern character, then $\beta_t(Y^{\rm ab})$ corresponds to $[l_{Y^{\rm ab}}]$. 

When $\vectnu$ is nonempty, $\Gamma_{g,\vectnu}^{\rm ab}$ may have torsion elements but its free part is still $\ZZ^{2g}$ and generated by $A_i^{\rm ab}, B_i^{\rm ab}$ as before. Rationally, $\beta_a^\sigma:\Gamma_{g,\vectnu}^{\rm ab}\rightarrow K_1(C^*_r(\Gamma_{g,\vectnu},\sigma))$ still gives an isomorphism, so $K_1(C^*_r(\Gamma_g,\sigma))$ is again generated by $[L_{A_i}^\sigma],[L_{B_i}^\sigma]$. In particular, $[L_Y^\sigma]$ depends only on $Y^{\rm ab}$. The Baum--Connes map is 
$$\mu_\theta=\mu_\theta^{\Gamma_{g,\vectnu}}:K_1^{\rm orb}(\Sigma_{g,\vectnu})=K_1^G(\Sigma_{g'})\rightarrow K_1(C^*_r(\Gamma_g,\sigma)),$$
and the homomorphism $\beta_t:\Gamma_{g,\vectnu}^{\rm ab}\rightarrow K_1^{\rm orb}(\Sigma_{g,\vectnu})$ is such that Eq.\ \eqref{valetteidentity} holds. In particular, $\beta_t$ vanishes on the torsion elements of $\Gamma_{g,\vectnu}^{\rm ab}$ \cite{Valette}. Using the Baum--Connes Chern character \cite{BC88} or delocalised equivariant homology \cite{Marcolli99}, we may identify 
$K_1^{\rm orb}(\Sigma_{g,\vectnu})\cong H_1(\Sigma_{g,\vectnu})\cong \ZZ^{2g}$ with the generating cycles $l_{A_i}, l_{B_i}$ as before. Note that torsion elements of $\Gamma_{g,\vectnu}^{\rm ab}$ such as $C_j$ do not contribute any nontrivial ($K$)-cycles.

With these descriptions, we can now state the effect of the noncommutative T-duality map
$$T_{\Sigma_{g,\vectnu}}: K^{-1}_{\rm orb}(\Sigma_{g,\vectnu})\rightarrow K_1(C^*_r(\Gamma_{g,\vectnu},\sigma)),$$
defined as $\mu_\theta$ composed with Poincar\'{e} duality, as follows.

\begin{proposition}\label{prop:genodd}
Let $[U]\in K^{-1}_{\rm orb}(\Sigma_{g,\vectnu})$ have Chern character whose Poincar\'{e} dual is $[l_{Y^{\rm ab}}]$, then $T_{\Sigma_{g,\vectnu}}([U])=[L_Y^\sigma]\in K_1(C^*_r(\Gamma_{g,\vectnu},\sigma))$. \label{K1Tdual}
\end{proposition}
Note that if $[U]$ is nontrivial, $Y$ is necessarily torsion free. Also, the intersection pairing qof cycles is such that $[l_{A_i}]\# [l_{B_j}]=\delta_{ij}$ so that the Poincar\'{e} duals can be described explicitly. For example, the Poincar\'{e} dual of $[l_{A_i}]$ evaluates to 1 on $[l_{B_i}]$ and kills the other generating cycles; e.g.\ this can be seen from Fig.\ \ref{fig:intersection}.

\section{Bulk-boundary maps}\label{section:extensions}
\subsection{Preliminaries: UCT and extensions}
The classical index theorem linking the Fredholm index of a Toeplitz operator with the winding number of its symbol may be understood in terms of $K$-theory and extensions as follows. We think of $C(\TT)\cong C^*_r(\ZZ)$ acting on $l^2(\ZZ)\cong L^2(\TT)$ so that $C(\TT)$ is generated by translations $L_n,$ $n\in\ZZ$ whose Fourier transforms are multiplication by $e^{{\rm i} n\theta}$. When truncated to Hardy space thought of as $l^2(\NN)$, the translation operator $L_1$ becomes a unilateral shift $\wt{L}_1$, and acquires a dimension 1 cokernel (the subspace for the boundary $n=0$). It is a Fredholm operator $\wt{L}_1=T_f$ with invertible symbol $f=e^{{\rm i}\theta}$ and index$(T_f)=-1=-{\rm Wind}(f)$. The Toeplitz algebra $\cT$ generated by $\wt{L}_1$ is a non-split extension
\beq
1\longrightarrow \cK \longrightarrow \cT\overset{\small \rm symbol}{\longrightarrow} C(\TT)\longrightarrow 1.\label{toeplitz}
\eeq
with kernel $\cK$ the compact operators on $l^2(\NN)$. In the seminal work of Brown--Douglas--Fillmore (BDF) \cite{BDF}, the extension theory of $C^*$-algebras $\cA$ by $\cK$ (denoted ${\rm Ext}(C(X))$ up to a certain notion of equivalence) was shown to be related to $K$-homology in the sense that ${\rm Ext}(C(X))\cong K_1(X)$ (at least when $X$ is a CW complex). From this point of view, Eq.\ \eqref{toeplitz} defines the generating element of $K_1(\TT)$, and the analytic index pairing $K_1(\TT)\times K^{-1}(\TT)\rightarrow K_0(\cK)\cong \ZZ$ taking $([\cT],[f])\mapsto {\rm Index}(T_f)$ realises the topological winding number/index of the symbol $f$ (up to a sign).

As explained in the latter part of this section, we will need to consider extensions of $C^*_r(\Gamma, \sigma)$ by $C(\TT)$ (actually its stabilisation), where $\Gamma=\Gamma_{g,\vectnu}$, in order to model bulk-boundary maps. Such extensions may be studied using the universal coefficient theorem (UCT) due to Rosenberg--Schochet \cite{Ros-Schoc}, and vastly generalises the BDF theory. 
By Corollary 7.2 in  \cite{Ros-Schoc}, $\cK \rtimes \Gamma \cong C^*_r(\Gamma, \sigma)\otimes\cK$ 
satisfies UCT, so that one has a short exact sequence,
\begin{align*}
0\to &{\rm Ext}^1_\ZZ(K_*(C^*_r(\Gamma, \sigma)), K_*(C(\TT))) \to KK^1(C^*_r(\Gamma, \sigma), C(\TT)) \to\\
&\to {\rm Hom}(K_1(C^*_r(\Gamma, \sigma)), K_0(C(\TT))) \oplus {\rm Hom}(K_0(C^*_r(\Gamma, \sigma)), K_1(C(\TT)))
\to 0.
\end{align*}
Since $K_*(C^*_r(\Gamma, \sigma))$ for $\Gamma=\Gamma_{g,\vectnu}$ and $K_*(C(\TT))$ are free abelian groups, ${\rm Ext}^1_\ZZ=0$, therefore
\begin{align*}
KK^1(C^*_r(\Gamma, \sigma), C(\TT))& \cong \\
  {\rm Hom}(K_1(C^*_r(\Gamma, \sigma)), & K_0(C(\TT))) \oplus {\rm Hom}(K_0(C^*_r(\Gamma, \sigma)), K_1(C(\TT))).
\end{align*}
This shows that any element $\alpha \oplus \partial$ of the RHS above determines a unique extension class \cite{Kas80,Blackadar}
$$
0 \to C(\TT) \otimes \cK \to E(\alpha \oplus \partial) \to C^*_r(\Gamma, \sigma) \to 0
$$
giving rise to the 6-term exact sequence in K-theory, with boundary maps
\begin{align*}
\alpha : K_1(C^*_r(\Gamma, \sigma)) \to K_0(C(\TT))\\
\partial : K_0(C^*_r(\Gamma, \sigma)) \to K_1(C(\TT)).
\end{align*}

\subsection{Bulk-boundary maps in Euclidean space}
In \cite{Kellendonk1,PSB} a model for bulk-boundary maps was introduced, in which a bulk $C^*$-algebra was extended by a boundary $C^*$-algebra, and the resulting boundary homomorphisms in $K$-theory taken to be bulk-boundary maps. This was applied successfully to prove equality of bulk and boundary conductivities in the physical context of the Integer quantum Hall effect. In the cases studied there, the bulk algebra is generated by a lattice $\ZZ^2$ of Euclidean magnetic translation symmetries generated by elements $A, B$, while the boundary symmetries comprised only the subgroup $\ZZ_A$ generated by $A$ which translated along the physical codimension-1 boundary (a Euclidean line containing a $\ZZ_A$ orbit which partitions Euclidean space into the ``bulk'' on one side and the ``vacuum'' on the other). The extension was taken to be a Toeplitz-like extension (in the sense of Pimsner--Voiculescu \cite{Pimsner}) with the effect of imposing boundary conditions on the bulk translations operators. Then the bulk-boundary map was (ignoring the modelling of disorder)
$$\partial_{\rm PV}:K_0(C^*_r(\ZZ^2,\sigma))\rightarrow K_1(C^*_r(\ZZ_A))=K_1(C(\TT)),$$
which mapped the class of the Rieffel projection $[\cP_\sigma]$ to the generator of $K_1(C(\TT))$ and mapped the trivial projection [{\bf 1}] to zero. Similarly, under $\alpha_{\rm PV}:K_1(C^*_r(\ZZ^2,\sigma))\rightarrow K_0(C^*_r(\ZZ_A))$, the class of the unitary $[L_B]$ maps to the generator whereas $[L_A]$ maps to zero. An explicit analysis of $\partial_{\rm PV}$ as a Kasparov product with the class of the above extension in $KK^1(C^*_r(\ZZ^2,\sigma),C(\TT))$ was carried out in \cite{Bourne}.

\subsection{Bulk-boundary maps in the hyperbolic plane}
For our hyperbolic plane generalisation, the bulk-algebra is taken to be $C^*_r(\Gamma,\sigma)$ as discussed in Section \ref{section:IQHE}, and we need a sensible notion of a ``boundary'' in $\HH$ and translations therein. A natural choice is to take the subgroup $\ZZ_X$ generated by some hyperbolic element (necessarily non-torsion) $X\in\Gamma$, and the boundary algebra to be $C^*_r(\ZZ_X)\cong C(\TT)$. The geometric meaning is as follows \cite{Iversen,Beardon}. 

Any hyperbolic transformation $X$ of $\HH$ has two idealised fixed points at infinity (they are two points on the boundary circle in the Poincar\'{e} disc model of $\HH$). There is a unique geodesic (hyperbolic straight line), called the \emph{axis} of $X$, connecting these fixed points. A hypercycle for $X$ comprises the points in $\HH$ which are on one side of and a fixed hyperbolic distance away from the geodesic. Thus all the hypercycles for $X$ only intersect (eventually) at the two fixed points --- a manifestation of the non-Euclidean geometry. The orbit of a given point in $\HH$ under $\ZZ_X$ is contained in a hypercycle for $X$ through that point, and will serve as a boundary partitioning $\HH$ into two sides (Fig.\ \ref{fig:hypercycle}). 

\begin{figure}[h]
    \centering
    \includegraphics[width=0.6\textwidth]{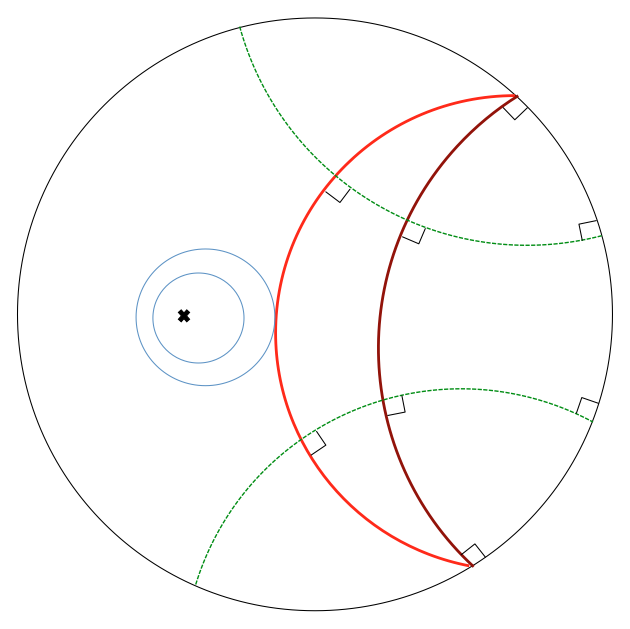}
    \caption{In the Poincar\'{e} disk model of the hyperbolic plane, geodesics (green) are arcs of (Euclidean) circles that orthogonally intersect the boundary circle at infinity. A hyperbolic transformation effects translations along hypercycles (red) connecting its two idealised fixed points at infinity. One such hypercycle is a geodesic (dark red). An elliptic transformation effects ``rotations'' about some fixed point $\times$, and each of its orbits lies in a hyperbolic circle (blue).}
    \label{fig:hypercycle}
\end{figure}

Note that such a ``boundary hypercycle'' is homeomorphic to $\RR$, and in the quotient $\HH/\Gamma=\Sigma$ it becomes a cycle $l_X:S^1=\RR/\ZZ_X\rightarrow\Sigma$. 
Recall that under the identification $\ZZ^{2g}\cong H_1(\Sigma)$, the homology classes are labelled by $[l_{X^{\rm ab}}]$ with $X^{\rm ab}\in\Gamma^{\rm ab}$, and we see that $[l_X]=[l_{X^{\rm ab}}]$. The intersection pairing is $[l_{A_i}]\# [l_{B_j}]=\delta_{ij}$, which means that we can interpret $B_i$ as a translation transverse to boundaries generated along $A_i$, while for $j\neq i$, the translations $A_j, B_j$ are \emph{not} transversal (their orbits cross the boundary an equal number of times in each direction). The geometry of intersections of hypercycles is illustrated for the special case of $\Gamma_{g=2}$ in Fig.\ \ref{fig:intersection}.

\begin{figure}[h]
    \centering
    \includegraphics[width=0.7\textwidth]{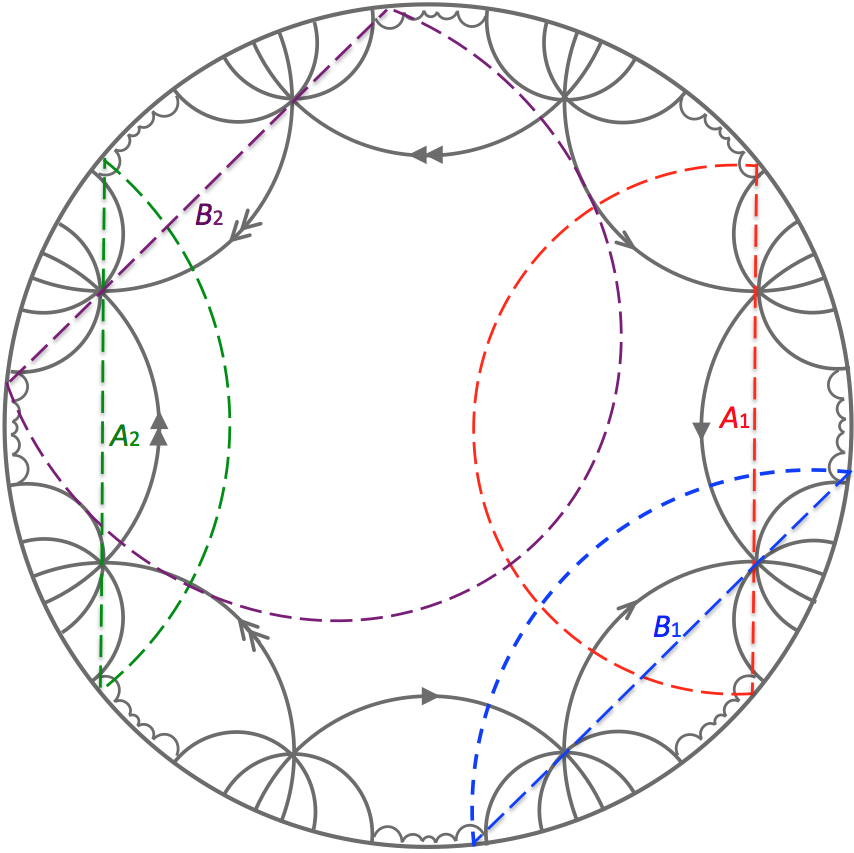}
    \caption{A tiling of the Poincar\'{e} disk by fundamental polygons for $\Gamma_g$ with $g=2$ (edges near the boundary circle are omitted to reduce clutter). The vertices of the polygons constitute a Cayley graph for $\Gamma_g$. The edges of the central octagon are pairwise identified according to the arrows to form the Riemann surface $\Sigma_g$. For each of the hyperbolic generators $A_1, B_1, A_2, B_2$ of $\Gamma_2$, two hypercycles are drawn (dashed lines). We see that $(A_i, B_i)$ are transversal to each other for $i=1$ and also $i=2$, but the other pairs are not. For instance, although the large $B_2$ hypercycle intersects the large $A_1$ hypercycle twice in total, their \emph{signed} intersection number is $0$.}
    \label{fig:intersection}
\end{figure}

Recall that $K_0(C^*_r(\Gamma))\cong \ZZ[\cP_\sigma]\oplus {\rm ker}\,\tau_c$ where we had distinguished $\ZZ[\cP_\sigma]$ as the classes that contribute to Hall conductance in Proposition \ref{BC} and Remark \ref{rem:cohomologous}. Also, we explained in Section \ref{section:NCTdualityodd} that $K_1(C^*_r(\Gamma,\sigma))\cong\ZZ^{2g}$ with all classes realisable by $[L_{Y^{\rm ab}}^\sigma]$ for some $Y^{\rm ab}\in\Gamma^{\rm ab}$. Guided by the geometry of the hypercyclic boundary defined by $X$ and the bulk-boundary map in the Euclidean case, we will define 
\beq
\partial_X  ([\cP_\sigma]) = [\zeta], \qquad \partial_X({\rm ker}\,\tau_c)=0, \label{boundarymap1}
\eeq
where $[\zeta]$ is the generator of $K_1(C(\TT))\cong K^{-1}(\TT)\cong\ZZ$, and
\beq
\alpha_X([L_Y^\sigma]) =\alpha_X([L_{Y^{\rm ab}}^\sigma])= ([l_{Y^{\rm ab}}]\#[l_{X^{\rm ab}}])\cdot[{\bf 1}].\label{boundarymap2}
\eeq
The intuition behind Eq.\ \eqref{boundarymap2} is that each translation transverse to the boundary gets modified to a ``half-translation'' and leaves behind a zero mode. This is analogous to the interpretation of the classical Fredholm index of the unilateral shift operator on $\NN$. Eq.\ \eqref{boundarymap1} generalises the Euclidean case, with the only change being that the extra non-trivial classes in $K_0(C^*_r(\Gamma,\sigma))$, which are due to the conical singularities $p_j$, are mapped to zero under $\partial_X$. This is reasonable since the boundary hypercycle is not generally preserved by the elliptic symmetries $C_j$ so it cannot ``see'' the $K$-theory classes arising from them.

Through the UCT, the combination $\alpha_X\oplus\partial_X$ specifies (albeit abstractly) the class of an extension of $C^*_r(\Gamma_{g,\vectnu},\sigma)$ by $C^*_r(\ZZ_X)\cong C(\TT)$, generalising the Toeplitz-like extension used in the Euclidean case. The sign ambiguity in $[\zeta]$ corresponds to the choice of side of the boundary to take as the bulk; notice that the induced orientation on the boundary depends on this choice.

\subsection{T-duality simplifies the bulk-boundary correspondence}\label{section:dualitytrivializes}
A hypercyclic boundary containing an orbit of $\ZZ_X$ gives rise to a cycle $l_X:S^1\rightarrow\Sigma$, and there is a pullback $l_X^*$ in $K$-theory which we can think of as a ``restriction map'' to the immersed image of $l_X$. When we regard the boundary as $\RR\simeq E\ZZ_X$, then the quotient under $\ZZ_X$ is $S^1=\RR/\ZZ_X\simeq B\ZZ_X$, and we can apply T-duality \eqref{circleTduality} to this circle.

\begin{theorem}\label{BBCsimplify}
The following diagram commutes for $\bullet=0,1$,
\beq
\xymatrix{
K^\bullet_{\rm orb}(\Sigma)  \ar[d]^{l_X^*} \ar[r]^{\sim}_{T_\Sigma} & K_\bullet(C^*_r(\Gamma, \sigma)) \ar[d]^{\partial_X} \\
K^\bullet(S^1) \ar[r]^{\sim}_{T_{\rm circle}} & K^{\bullet-1}(\TT)  } .\nonumber
\eeq
\end{theorem}
\begin{proof}
We check this by computing the maps on generators of $K^\bullet_{\rm orb}(\Sigma)$. For $\bullet=0$, since $S^1$ is one dimensional, it is clear that only the rank invariant survives under $l_X^*$. We compute
\beq
T_{\rm circle}\circ l_X^*([{\bf 1}]_{\Sigma})=T_{\rm circle}([{\bf 1}]_{S^1})=[\zeta]\nonumber
\eeq
By Corollaries \ref{cor:evengenformula1}, \ref{cor:evengenformula2}, and the formula for $\partial_X$ in Eq.\ \eqref{boundarymap1}, we also have
\beq
\partial_X\circ T_\Sigma([{\bf 1}]_{\Sigma})=\partial_X([\cP_\sigma])=[\zeta].\nonumber
\eeq
As for the other generators $[\wt{\cL}],[\wt{\cL}_{\chi_j}]$ of $K^0_{\rm orb}(\Sigma)$, they are mapped to zero by $l_X^*$ (and thus by $T_{\rm circle}\circ l_X^*$), while $T_\Sigma$ takes them to ${\rm ker}\,\tau_c$ which vanishes under $\partial_X$.

For $\bullet=1$, let $[U]\in K^{-1}_{\rm orb}(\Sigma)$ have Chern character ${\rm Ch}([U])$ Poincar\'{e} dual to $[l_{Y^{\rm ab}}]$, and $[W]$ be a generator of $K^{-1}(S^1)$. We compute
\begin{align*}
\langle{\rm Ch}(l_X^*[U]),[S^1]\rangle&=\langle{l_X^*\rm Ch}([U]),[S^1]\rangle \\ 
&=\langle {\rm PD}([l_{Y^{\rm ab}}]), (l_X)_*([S^1]) \rangle \\
&=\langle {\rm PD}([l_{Y^{\rm ab}}]), [l_{X^{\rm ab}}] \rangle \\
&=[l_{Y^{\rm ab}}]\#[l_{X^{\rm ab}}] \\
&=\langle[l_{Y^{\rm ab}}]\#[l_{X^{\rm ab}}]\cdot [W] , [S^1]\rangle,
\end{align*}
which combined with Eq.\ \eqref{circleTduality} gives
\beq
T_{\rm circle}\circ l_X^*([U])=T_{\rm circle}(([l_{X^{\rm ab}}]\#[l_{Y^{\rm ab}}])\cdot[W])=[l_{X^{\rm ab}}]\#[l_{Y^{\rm ab}}]\cdot[{\bf 1}]\nonumber
\eeq
On the other hand, using Proposition \ref{prop:genodd} we also get
\beq
\alpha_X\circ T_\Sigma([U])=\alpha_X([L_Y^\sigma])=[l_{X^{\rm ab}}]\#[l_{Y^{\rm ab}}]\cdot[{\bf 1}].\nonumber
\eeq
\end{proof}

To summarise, the relatively complicated and abstract boundary map,
\beq
\alpha_X\oplus\partial_X \colon K_*(C^*_r(\Gamma, \sigma))  \to K^{*-1}(\TT) \nonumber
\eeq
is equivalent to the conceptually simple geometric restriction-to-boundary maps,
\beq
l_X^* \colon K^\bullet_{\rm orb}(\Sigma)  \to K^\bullet(S^1), \qquad \bullet=0,1. \nonumber
\eeq
Reversing the argument, this says that the extension defined by $\alpha_X\oplus\partial_X$ indeed correctly captures the geometry of the bulk-boundary relationship.

Theorem \ref{BBCsimplify} generalises to the hyperbolic plane geometry, the Euclidean space result \cite{MT1} and the Nil and Solv geometry results \cite{HMT1,HMT2}.

\subsection{Cyclic cohomology, Hall conductance and boundary conductance}\label{section:cyclic}

Since the group $\Gamma=\Gamma_{g,\vectnu}$ is a cocompact discrete subgroup of $\slr$, it has the {\em rapid decrease} (RD) property, see \cite{Harpe} and chapter 8 in \cite{Valette-book}. It follows that there is a smooth subalgebra 
 $C^*_r(\Gamma, \sigma)^\infty \hookrightarrow C^*_r(\Gamma_g, \sigma)$ inducing an isomorphism in K-theory. The periodic cyclic cohomology $HP^{\rm even}(C^*_r(\Gamma_g, \sigma)^\infty)$ includes $[\tau]$ and $[\tau_K]$ where $\tau$ is the von Neumann trace on $C^*_r(\Gamma, \sigma)^\infty$ and $\tau_K$ is the conductance cocycle (see \cite{CHMM,Marcolli06} for details and its meaning as a higher genus Kubo formula), 
\beq
\tau_K(f_1,f_2,f_3) = \frac{1}{g}\sum_{i=1}^g
\tau(f_1(\delta_i(f_2)\delta_{i+g}(f_3) - \delta_{i+g}(f_2)\delta_i(f_3)))\nonumber 
\eeq 
for $f_1, f_2, f_3 \in C^*_r(\Gamma, \sigma)^\infty$. Here the derivations $$\delta_i(f)(\gamma)=a_i(\gamma)f(\gamma), \qquad 
\delta_{i+g}(f)(\gamma) = b_i(\gamma) f(\gamma),$$ where $a_i, b_i$ are the group 1-cocycles on $\Gamma$ corresponding 
to the generators $A_i, B_i$.

The Kubo conductance cocycle $\tau_K$ is actually cohomologous to $\tau_c$ (\cite{CHMM}, Theorem 4.1 of \cite{Marcolli01}), which facilitates the calculation of the range of $\tau_c$ on $K_0(C^*_r(\Gamma,\sigma))$ through a higher twisted index theorem (Eq.\ \eqref{higherindexformula}) --- the range is $\phi\ZZ$ where $\phi=2(g-1)+\sum_{j=1}^r\frac{1}{\nu_j}\in\QQ$ \cite{Marcolli01}, with the value $\phi$ achieved by $[\cP_\sigma]$. Thus $\phi\ZZ$ is the range of possible values of the Hall conductance.

For conductance along the boundary, we also need to know that the periodic cyclic cohomology of the smooth functions on the circle $C^\infty(\TT)$ is
\beq	
H^1(S^1) = HP^{\rm odd}(C^\infty(\TT)) = \CC[\tau_w] \nonumber
\eeq
where $\tau_w(U)$ evaluates the winding number of a unitary function $U$ on $\TT$. Just as in the Euclidean case \cite{Kellendonk1}, $\tau_K$ should be obtained from the boundary conductance 1-cocycle $\tau_{\rm bd}$ under the map $\breve{\partial}_X$ in $HP$ dual to $\partial_X$, that is,
\beq
\tau_{\rm bd}(\partial_X([\cP]))=(\breve{\partial}_X\tau_{\rm bd})([\cP])\equiv\tau_K([\cP]) ,\qquad [\cP]\in K_0(C^*_r(\Gamma,\sigma)).\label{dualitypairing}
\eeq
It suffices to consider $[\cP]=[\cP_\sigma]$, whence Eq.\ \eqref{dualitypairing} reads
\beq
\tau_{\rm bd}([\zeta])=\tau_K([\cP_\sigma])=\phi.\nonumber
\eeq
Since $\tau_w([\zeta])=1$, we have $\tau_{\rm bd}=\phi\tau_w$ as the boundary conductance 1-cocycle.

Notice that there is a geometric factor of $\phi$ in $\tau_{\rm bd}$ which was simply $1$ in the Euclidean (genus 1) case. Intuitively, the way in which the boundary is embedded in $\HH$ depends strongly on the data of $g,\vectnu$ in $\G$, so that the boundary does ``feel'' the bulk geometry.

\section{Modelling of disorder with crossed products}\label{section:generalisations}

\subsection{Cyclic cocycles and Hall conductance in the presence of disorder} \label{sec:cyclicHall}
Following \cite{Bellissard, Prodan}, one can model the effect of disorder by using a compact space $\Omega$ of disorder configurations on which $\Gamma=\G$ acts. Instead of $C^*_r(\Gamma, \sigma)$, we need to use the ($\sigma$-twisted) reduced crossed product algebra $C(\Omega)\rtimes_{\sigma, r} \Gamma$.

We wish to discuss cyclic cocycles on a
smooth subalgebra of the (reduced) crossed product algebra $C(\Omega)\rtimes_{\sigma} \Gamma$, denoted 
$(C(\Omega)\rtimes_{\sigma} \Gamma)^\infty$. We need to assume besides a minimal action of $\Gamma$
on $\Omega$, that the $\Gamma$-action has an invariant probability measure $\mu$. This is possible when $\Gamma$ acts on $\Omega$
via an amenable quotient. We recall that a locally compact Hausdorff group is said to be \emph{amenable} if it admits a left (or right) invariant \emph{mean}, which for discrete groups simplifies to having an invariant finitely additive probability measure. Amenable groups are closed under many operations such as taking subgroups, quotients and group extensions. However, the free group on two generators is not amenable, and so $\Gamma$ is not amenable.

Let us examine the hypothesis that $\Gamma$ acts via an amenable quotient. 
There is a surjective homomorphism $\Gamma=\G \twoheadrightarrow F_{g}$ onto the free group on $g$ generators, mapping the generators $C_j \mapsto e$ and $B_i\mapsto e$. 
Now any amenable group $\mathscr{A}$ with $g$ generators is a quotient of $F_g$, therefore there is a surjective homomorphism $\Gamma \twoheadrightarrow \mathscr{A}$ and any
minimal action of $\mathscr{A}$ on $\Omega$ satisfies the hypothesis. Finally, the constructive proof of the remarkable Corollary 1.5 in \cite{Hjorth} asserts that every countable discrete group acts freely and minimally on a Cantor set, so there are many such examples.

A cyclic 2-cocycle on $(C(\Omega)\rtimes_{\sigma} \Gamma)^\infty$ can be defined from a group 2-cocycle $c$ as follows.
Let $f_j(x, \gamma) \in (C(\Omega)\rtimes_{\sigma} \Gamma)^\infty$ for $j=0,1,2$. Then 
\begin{align*}
{\rm tr}_{c, \mu}(f_0,f_1,f_2)&=\sum_{\gamma_0\gamma_1\gamma_2=1} 
\int_\Omega  \big[f_0(x, \gamma_0)f_1(\gamma_0^{-1}x, \gamma_1)f_2({(\gamma_0\gamma_1)}^{-1}x, \gamma_2) 
\\
&\qquad\qquad\qquad\qquad\qquad\qquad c(1, \gamma_1, \gamma_1\gamma_2) \sigma(\gamma_1, \gamma_2)\big]\, d\mu(x)
\end{align*}
defines a cyclic 2-cocycle on $(C(\Omega)\rtimes_{\sigma} \Gamma)^\infty$.

Let $c$ denote the {\em area cocycle} defined in Section \ref{sect:FM}. Recall that we had $\tau_c([\cP_\sigma])=\phi$ where $\phi=2(g-1)+(r-\nu)\in\QQ$ is the orbifold Euler characteristic of $\Sigma=\Sigma_{g,\vectnu}$. The natural
inclusion
                $C^{*}_{r}(\Gamma, \sigma)=\CC\rtimes_\sigma \Gamma\hookrightarrow C(\Omega)\rtimes_\sigma\Gamma\,$ takes the (virtual)
 projection $\cP_\sigma$ to a (virtual) projection $\sP_\sigma$ in a matrix algebra over $C(\Omega)\rtimes_\sigma\Gamma$, and $\sP_\sigma$ does not depend on $x\in\Omega$. So we have
\begin{equation}
{\rm tr}_{c, \mu}([\sP_\sigma])\equiv{\rm tr}_{c, \mu}(\sP_\sigma, \sP_\sigma, \sP_\sigma) =\tau_c([\cP_\sigma]) \int_\Omega d\mu(x)=\phi.\label{disorderHallcontribution}
\end{equation}
Note that $[\sP_\sigma]$ is then a non-torsion class.

We can understand $\sP_\sigma$ more geometrically through its T-dual, constructed in the following way. Let $M= \Omega \times_{\Gamma} \hyp, \,$ which is an orbifold fibre bundle over $\Sigma=\Sigma_{g,\vectnu}$. 
Let $\cE$ be a (orbifold) vector bundle over $M$, then there is a twisted foliation index theorem \cite{BMprogress} generalizing the index theorem in \cite{BM18}
\beq
{\rm tr}_{c, \mu}(\mu_\theta([\dirac_M\otimes \mathcal{E}])) = \frac{1}{2\pi}\int_{M} e^\B \wedge {\rm Ch}(\mathcal{E})\wedge
\omega_c \label{twistedfoliation}
\eeq
where $\omega_c $ is the hyperbolic volume 2-form on $\Sigma$ pulled back to $M$, corresponding to the area 2-cocycle $c$ on 
$\Gamma$. In analogy to the Dirac operator playing the role of a fundamental class, $\dirac_M$ in Eq.\ \eqref{twistedfoliation} is the Dirac operator on the orbifold $\Sigma=\Sigma_{g,\vectnu}$ lifted to $M$, which is elliptic along the leaves of the foliation, and $\dirac_M\otimes\cE$ is its $\cE$-twisted version. This generalises $\dirac_\cE^+$ in Section \ref{sect:fuchsianTduality}. Also, $\mu_\theta$ is the twisted Baum--Connes map with coefficients $C(\Omega)$, and we define T-duality to be the map
$$
T:K^0_{\rm orb}(M)\ni [\cE]\mapsto \mu_\theta([\dirac_M\otimes\cE])\in K_0(C(\Omega)\rtimes_\sigma\Gamma).
$$

The integral in Eq.\ \eqref{twistedfoliation} simplifies to
$$
\frac{1}{2\pi}\int_{M} e^\B \wedge {\rm Ch}(\mathcal{E})\wedge
\omega_c =\frac{1}{2\pi} \int_{M} \text{rank}(\cE) \omega_c .
$$
Since the action of $\Gamma$ on $\Omega$ is minimal, so that $M$ is connected (Lemma 3, \cite{BO07}),  the integer valued function $\text{rank}(\cE)$
is a constant, therefore there is a further simplification
\begin{align*}
\frac{1}{2\pi}\int_{M} \text{rank}(\cE) \omega_c  = \text{rank}(\cE) \frac{1}{2\pi}\int_{M} \omega_c & =  \text{rank}(\cE)\mu(\Omega) \frac{1}{2\pi}\int_{\Sigma} \omega_c \\
&= \phi\, \text{rank}(\cE)=\text{rank}(\cE){\rm tr}_{c, \mu}([\sP_\sigma]).
\end{align*}
This means that up to a term in the kernel of ${\rm tr}_{c, \mu}$, the T-duality map $T:K^0_{\rm orb}(M)\rightarrow K_0(C(\Omega)\rtimes_\sigma\Gamma)$ takes $[{\bf 1}_M]$ (the class of the trivial line bundle over $M$) to $[\sP_\sigma]$, whereas $\wt{K}^0_{\rm orb}(M)$ is mapped to the kernel of ${\rm tr}_{c, \mu}$, i.e.\
\beq
    T([\mathcal{E}])={\rm rank}(\mathcal{E})[\sP_\sigma]+\mathcal{C},\qquad \mathcal{C}\in {\rm ker}\,\,{\rm tr}_{c, \mu}.\label{disorderTaction}
\eeq

As in Section \ref{section:cyclic}, we can define the disorder-averaged Kubo conductivity cocycle ${\rm tr}_{K, \mu}$ and show that it is cohomologous to ${\rm tr}_{c, \mu}$. We see that the subgroup $\ZZ[\sP_\sigma]\subset K_0(C(\Omega)_\sigma\Gamma))$ may be interpreted as that which contributes to the Hall conductance in the presence of disorder.

\subsubsection{Boundary conductivity cocycles in the presence of disorder}
For the boundary algebra, we use $C(\Omega)\rtimes\ZZ_X$, noting that $\sigma$ is trivial on $\ZZ_X$. It is known (cf.\ \cite{Putnam}) that 
\begin{equation}
K_0(C(\Omega)) = C(\Omega, \Z), \qquad K_1(C(\Omega)) =0.\label{cantorKgroups} 
\end{equation}
From the Pimsner--Voiculescu (or Kasparov spectral) sequence, we deduce that 
$$K_1(C(\Omega)\rtimes\ZZ_X)\cong K_0(C(\Omega))^{\ZZ_X}=C(\Omega,\ZZ)^{\ZZ_X},$$
i.e.\ the $\ZZ_X$-invariant part of $C(\Omega,\ZZ)$. In particular, the class $[\xi]\in K_1(C(\Omega)\rtimes\ZZ_X)$ induced from the unitary $\zeta\in C^*_r(\ZZ_X)$ under inclusion of scalars $\CC\hookrightarrow C(\Omega)$, corresponds to the constant invariant function $\Omega\mapsto 1$.

Let us define the boundary map $\partial_X:K_0(C(\Omega)\rtimes_\sigma\Gamma)\rightarrow K_1(C(\Omega)\rtimes\ZZ_X)$ to be
$$\partial_X  ([\sP_\sigma]) = [\zeta], \qquad \partial_X({\rm ker}\,{\rm tr}_{c, \mu})=0,$$
generalising Eq.\ \eqref{boundarymap1}.

\subsubsection{Bulk-boundary map is T-dual of restriction map in presence of disorder}
\begin{proposition}\label{BBCsimplifydisordered}
The following diagram commutes,
\beq
\xymatrix{
K^0_{orb}(M)  \ar[d]^{\iota_X^*} \ar[r]^{\sim\qquad}_{T\qquad} & K_0(C(\Omega)\rtimes_{\sigma} \Gamma) \ar[d]^{\partial_X} \\
K^0(M_1) \ar[r]^{\sim\qquad}_{T_1\qquad} & K_1(C(\Omega)\rtimes  \ZZ_X)  }.  \nonumber
\eeq
Here,  
$M= \Omega \times_{\Gamma} \hyp, \,$
$M_1=\Omega\times_{\ZZ_X}\RR$ with $\RR$ a hypercycle for $X$ whose inclusion in $\HH$ induces $M_1\stackrel{\iota_X}{\rightarrow}M$.
\end{proposition}

\begin{proof}

First, we note that a vector bundle $\mathcal{E}$ over $M$ has constant rank everywhere, since $M$ is connected.
Then the rank gives a splitting $K^0_{\rm orb}(M)=\widetilde{K}^0_{\rm orb}(M)\oplus\ZZ[{\bf 1}_M]$. 
Also, there is a T-duality isomorphism $T_1^{-1}:K^0(M_1)\cong K_1(C(\Omega)\rtimes\ZZ_X)\cong C(\Omega,\ZZ)^{\ZZ_X}$ by standard arguments; elements of $C(\Omega,\ZZ)^{\ZZ_X}$ are integer linear combinations of characteristic functions on $\ZZ_X$-invariant clopen subsets $S\subset\Omega$. In particular, the constant function $\Omega\mapsto 1$ corresponding to $[\xi]$ T-dualises to $[{\bf 1}_{M_1}]$ which generates $\ZZ[{\bf 1}_{M_1}]\subset K^0(M_1)$. Generally, a characteristic function on $S$ corresponds to a trivial line bundle over the subbundle of $M_1$ with fibre $S$. Since a non-zero $\mathcal{E}\rightarrow M$ is supported on all of $M$, the homomorphism $\iota_X^*$ lands on the subgroup $\ZZ[{\bf 1}_{M_1}]$, taking $[\mathcal{E}]$ to ${\rm rank}(\mathcal{E})[{\bf 1}_{M_1}]$. In summary,
\beq
    T_1\circ \iota_X^*([\mathcal{E}])={\rm rank}(\mathcal{E})[\xi].\label{tdualofrestriction}
\eeq
From Eq.\ \eqref{disorderTaction}, one deduces that
\beq
    \partial_X\circ T([\mathcal{E}])=\partial_X\big({\rm rank}(\mathcal{E})[\sP_\sigma]+\mathcal{C}\big)={\rm rank}(\mathcal{E})[\xi],\nonumber
\eeq
which together with \eqref{tdualofrestriction}, shows that the above diagram commutes. 

\end{proof}

A ``disorder-averaged'' cyclic 1-cocycle on $(C(\Omega)\rtimes \ZZ_X)^\infty$ is defined from a group 1-cocycle $a$ on $\ZZ_X$ as follows.
Let $f_j(x, \gamma) \in (C(\Omega)\rtimes \ZZ_X)^\infty$ for $j=0,1$. Then 
$$
{\rm tr}_{a, \mu}(f_0,f_1)=\sum_{\gamma_0\gamma_1=1} \int_\Omega f_0(x, \gamma_0)f_1(\gamma_0^{-1}x, \gamma_1)  d\mu(x) a(1, \gamma_1) 
$$
defines a cyclic 1-cocycle on $(C(\Omega)\rtimes \ZZ_X)^\infty$. Let $a$ be the group 1-cocycle on $\ZZ_X$ that represents the generator of $H^1(\ZZ_X,\ZZ)$, and gives rise to the winding number cyclic cocycle $\tau_w$ on $(C^*_r(\ZZ_X))^\infty$. Then $\phi\, a$ is a real-valued group cocycle on $\ZZ_X$.
Since 
$\tau_w([\zeta])=1$, we see that 
$${\rm tr}_{\phi a, \mu}([\xi])={\rm tr}_{\phi a, \mu}(\xi^{-1}, \xi)=\tau_w([\zeta])\int_\Omega d\mu(x)=\phi.$$ 

Together with Eq.\ \eqref{disorderHallcontribution} and ${\rm tr}_{K, \mu}\cong{\rm tr}_{c, \mu}$, we have
$$
{\rm tr}_{K, \mu}([\sP_\sigma])={\rm tr}_{c, \mu}(\sP_\sigma, \sP_\sigma, \sP_\sigma) = \phi={\rm tr}_{\phi a, \mu}([\xi])={\rm tr}_{\phi a, \mu}(\partial_X[\sP_\sigma]),
$$
so that taking the boundary conductance cocycle to be ${\rm tr}_{\phi a, \mu}$ yields the equality of bulk and boundary conductance.

\begin{remark}
To define an extension of $C(\Omega)\rtimes_\sigma\Gamma$ by $C(\Omega)\rtimes \ZZ_X$, we need (besides $\partial_X$ above) to also define the bulk-boundary map $\alpha_X$ in the other $K$-theory degree. Furthermore, it is not clear that $K_*(C(\Omega)\rtimes_\sigma\Gamma), K_*(C(\Omega)\rtimes \ZZ_X)$ are torsion-free, so there may be some further nonuniqueness of the extension due to the ${\rm Ext^1_\ZZ}$ term in the UCT (although this is not a problem if $\Gamma$ acts through the amenable group $\mathscr{A}=\ZZ^2$ or $\ZZ^3$, see \cite{BM18}). Nevertheless, these freedoms are independent of the boundary map $\partial_X$ that we are studying in this Section, so we leave open the specific choice of extension.
\end{remark}

\begin{remark}
For plateaux formation in quantum Hall effects, one should work in the regime of ``strong disorder'' or a ``gap of extended states'', which involves passing to noncommutative Sobolev spaces \cite{Bellissard,Prodan,PSB}. In the hyperbolic plane and for torsion-free $\Gamma=\Gamma_g$, this was studied for bulk phases in \cite{CHM}. The bulk-boundary correspondence in a gap of extended states remains a difficult problem even in the Euclidean case, and we intend to study this in a future work.
\end{remark}

\section{The time-reversal invariant topological insulator on the hyperbolic plane}\label{section:hyperbolicKM}
The following stable splitting lemma is useful for computing the complex and real $K$-theories of $\Sigma_g$:
\begin{lemma}\label{lemma:stablesplitting}
$\Sigma_g$ is stably homotopic to the wedge sum of $2g$ circles and a 2-sphere.
\end{lemma}
\begin{proof}
The 1-skeleton of $\Sigma_g$ is $\bigvee_{i=1}^{2g}S^1$, and the attaching map of its 2-cell is the product of commutators of inclusions $S^1\hookrightarrow S^1$ in accordance with the fundamental polygon. For the suspension $S\Sigma$ of $\Sigma_g$, the 3-cell is attached null-homotopically to the 2-skeleton $\bigvee_{i=1}^{2g}S^2$ since $\pi_2$ is abelian. Thus $S\Sigma\simeq (\bigvee_{i=1}^{2g}S^2)\vee S^3\simeq S(S^2\bigvee_{i=1}^{2g}S^1)$.
\end{proof}

A fermionic time-reversal symmetry $\mathsf{T}$ has the property that it is antiunitary and $\mathsf{T}^2=-1$. It is assumed to act pointwise in space, whether Euclidean or hyperbolic, and to commute with (magnetic) translations. This restricts the 2-cocycle $\sigma$ to be ${\rm O}(1)$ instead of ${\rm U}(1)$-valued, and a realistic situation under this condition is that of zero magnetic field, $\theta=0, \sigma\equiv 1$, rather than e.g.\ $\theta=\frac{1}{2}$ in Eq.\ \eqref{multiplier}. In the absence of time-reversal symmetry, the arguments in Section \ref{section:IQHE} led us to classify spectral projections of $\Gamma_g$-invariant Hamiltonians using $K_0(C^*_r(\Gamma_g))$ or equivalently $K_0(C^*_{r,\RR}(\Gamma_g)\otimes_\RR\CC)$, where $C^*_{r,\RR}(\Gamma_g)$ denotes the reduced group $C^*$-algebra of $\Gamma_g$ over the \emph{reals}. Both the operator $\mathsf{T}$ and the complex scalars $\CC$ commute with $C^*_{r,\RR}(\Gamma_g)$, but $\mathsf{T}{\rm i}=-{\rm i}\mathsf{T}$ means that there is an action of the quaternions $\mathbf{H}$ which the $\mathsf{T}$-symmetric Hamiltonians must be compatible with. The upshot is that we need to compute instead \cite{Thiang}
\beq
KO_0(C^*_\RR(\Gamma_g)\otimes_\RR\mathbf{H})\cong KO_4(C^*_\RR(\Gamma_g)),\label{KOgroup}
\eeq
with a degree-4 shift \cite{Kar}. By the real Baum--Connes isomorphism \cite{Baum} and Poincar\'{e} duality, this is 
$$KO_4(\Sigma_g)\cong KO^{-2}(\Sigma_g)\cong \wt{KO}^{-2}(S^0\vee\Sigma_g).$$
In $\wt{KO}$-theory, we can replace $S^0\vee\Sigma_g$ by a stably homotopic space, which by Lemma \ref{lemma:stablesplitting} is $S^0\vee S^2\bigvee_{i=1}^{2g}S^1$. Then
\begin{align*}
KO^\bullet(\Sigma_g)&\cong \wt{KO}^\bullet(S^0)\oplus\wt{KO}^\bullet(S^2)\bigoplus_{i=1}^{2g}\wt{KO}^\bullet(S^1)\\
&\cong KO^\bullet({\rm pt})\oplus KO^{\bullet-2}({\rm pt})\bigoplus_{i=1}^{2g}KO^{\bullet-1}({\rm pt}).
\end{align*}
\begin{definition}\label{def:hypKM}
The $K$-theoretic hyperbolic Kane--Mele invariant is the $\ZZ_2\subset KO^{-2}(\Sigma_g)$ corresponding to $\wt{KO}^{-2}(S^0)\cong KO^{-2}({\rm pt})\cong \ZZ_2$.
\end{definition}

Another method to compute $KO_\bullet(C^*_\RR(\Gamma_g))\cong KO_\bullet(\Sigma_g)$ is to use Poincar\'{e} duality $KO_\bullet(\Sigma_g)\cong KO^{2-\bullet}(\Sigma_g)$ and compute the latter using the Atiyah--Hirzebruch spectral sequence. The $E^2$ term is $H^p(\Sigma_g,KO^q({\rm pt}))$, so for $\bullet=4$ we have the terms
$$H^0(\Sigma_g,KO^{-2}({\rm pt}))\oplus H^2(\Sigma_g,KO^{-4}({\rm pt}))\cong\ZZ_2\oplus\ZZ$$
The $\ZZ_2$ indeed splits off in $KO^{-2}(\Sigma_g)$ as it arises from the inclusion of a point.

\begin{remark}
In the Euclidean case, the $\ZZ_2$ subgroup in $KO^{-2}(\TT^2)$ $=\wt{KO}^{-2}(S^0\vee\TT^2)$ corresponds under T-duality to $\wt{KR}^{-4}(\TT^2,\varsigma)\cong\ZZ_2$ \cite{MT2}, where the latter $\TT^2$ is the Brillouin torus equipped with the momentum reversal involution $k\mapsto -k$. This $\ZZ_2$ invariant was originally identified as the $K$-theoretic Kane--Mele invariant \cite{KM} in \cite{Kitaev}, and justifies the terminology of our hyperbolic generalisation in Definition \ref{def:hypKM}. Indeed $\wt{KR}^{-4}(\TT^2,\varsigma)$ was computed in \cite{Kitaev} using the Baum--Connes isomorphism, although it can also be computed directly using an equivariant stable splitting of $(\TT^2,\varsigma)$ \cite{FM}. In the hyperbolic case, we do not have momentum ``space'' in the usual sense, but Eq.\ \eqref{KOgroup} is still available.
\end{remark}

{\bf Towards fractional Kane--Mele indices.} When $\vectnu$ is nonempty, the computation of $KO^{-2}_{\rm orb}(\Sigma_{g,\vectnu})$ appears more difficult. Nevertheless, there is still a $\ZZ_2$-factor coming from inclusion of a point. In the absence of $\mathsf{T}$, the geometry of $\Sigma_{g,\vectnu}$ causes the pairing of the area cocycle with $K_0(C^*_r(\G,\sigma))$ to acquire a geometrical factor $\phi$ to become $\phi\ZZ\subset\RR$-valued. Typically, torsion $K$-theory classes pair trivially with cyclic cocycles. Nevertheless, in the presence of $\mathsf{T}$ and in the Euclidean case (as well as with disorder), modified cyclic cocycles were considered in \cite{Kellendonk4} which resulted in $\ZZ$-valued pairings with $KO_0(C(\Omega,{\mathbf H})\rtimes(\ZZ^2))$ well-defined modulo 2. We anticipate that the hyperbolic generalisation of \cite{Kellendonk4} along the lines of this paper will lead to fractional Kane--Mele ($\ZZ_2$) indices. For the bulk-boundary map, the extension theory and UCT in the real case was given in \cite{MR88}. We intend to also develop a fractional bulk-boundary correspondence of Kane--Mele type indices in a future work.

\bigskip

{\em Acknowledgements.}  VM thanks A.\ Carey, K. Hannabuss, M. Marcolli for past collaborations on QHE on the hyperbolic plane. GCT also thanks K.\ Gomi for some interesting discussions. This work was supported by the Australian Research Council via ARC Discovery Project grant DP150100008, Australian Laureate Fellowship  FL170100020 and DECRA DE170100149.

\appendix
\section{Appendix: Kasparov spectral sequence}
Let $\cA$ be a unital $C^*$-algebra with a $\Gamma_g$-action. According to Kasparov, section 6.10 in \cite{Kas88}, there is a spectral sequence $(E^r, d^r)$, generalising the PV-sequence, that converges to 
 $K_0(\cA\rtimes_{\sigma} \Gamma_g)$ with $E^2$ term equal to 
 \begin{equation}
 H_0(\Gamma_g, K_0(\cA)) \oplus H_1(\Gamma_g, K_1(\cA)) \oplus H_2(\Gamma_g, K_0(\cA)), \nonumber 
 \end{equation}
 where the factors are group homologies with coefficients in the induced $\Gamma_g$-module $K_i(\cA)$ \cite{Brown}. Applying Poincar\'{e} duality to the last term gives
 $$H_0(\Gamma_g, K_0(\cA)) \oplus H_1(\Gamma_g, K_1(\cA)) \oplus H^0(\Gamma_g, K_0(\cA)),$$ which simplifies to
 \beq\label{kasp:disorder}
 K_0(\cA)_{\Gamma_g} \oplus H_1(\Gamma_g, K_1(\cA)) \oplus K_0(\cA)^{\Gamma_g}
 \eeq
 where $K_0(\cA)_{\Gamma_g}$ denotes the coinvariants and $K_0(\cA)^{\Gamma_g}$ the invariants \cite{Brown}. 

 \subsection{Cantor disorder space with minimal action of $\Gamma_g$}
As an example, let $\cA=C(\Omega)$ with $\Omega$ a Cantor set. Suppose $\Omega$ is equipped with a mimimal action of $\Gamma_g$, and let $\sigma$
be a multiplier on $\Gamma_g$. Let us mention that there exist examples of minimal actions of $\Gamma_g$ on a Cantor set \cite{Hjorth}.

Using Eq.\ \eqref{cantorKgroups} and Eq.\ \eqref{kasp:disorder}, we see that the differentials in the Kasparov spectral sequence vanish in this special case, so that
\beq
K_0(C(\Omega) \rtimes_\sigma \Gamma_g)\cong C(\Omega, \Z)_{\Gamma_g} \oplus \Z,\nonumber
\eeq
where $C(\Omega, \Z)_{\Gamma_g} $ are the co-invariants under the $\Gamma_g$-action, 
and since $\Gamma_g$ acts minimally on $\Omega$, so $C(\Omega, \Z)^{\Gamma_g} \cong \ZZ$. More precisely,
\begin{itemize}
        \item    The natural inclusion
                $C^{*}_{r}(\Gamma_g, \sigma)\hookrightarrow C(\Omega)\rtimes_\sigma\Gamma_g\,$ takes $\cP_\sigma$ to $\sP_\sigma$, and $[\sP_\sigma]$ generates the $\Z$
factor in $K_0(C(\Omega)\rtimes_\sigma\Gamma_g)$.
                    \item   The inclusion $C(\Omega,\Z)_{\Gamma_g} \hookrightarrow K_0(C(\Omega)\rtimes_\sigma\Gamma_g)$
is induced by
                the inclusion $C(\Omega)\hookrightarrow C(\Omega)\rtimes_\sigma\Gamma_g$.
\end{itemize}
Thus we have
\beq
K_0(C(\Omega) \rtimes_\sigma \Gamma_g)\cong C(\Omega, \Z)_{\Gamma_g} \oplus \Z[\sP_\sigma].\nonumber
\eeq

\begin{remark}
In the above case with $\Gamma=\Gamma_g$, the kernel of ${\rm tr}_{c, \mu}$ may be identified with the subgroup $C(\Omega,\ZZ)_{\Gamma_g}$ of $K_0(C(\Omega)\rtimes_\sigma\Sigma_g)$ as follows. Elements of the latter subgroup has representative functions supported at the identity element of $\Sigma_g$. Such functions are killed by the derivations $\delta_j$ in the definition of ${\rm tr}_{c, \mu}$, so $C(\Omega,\ZZ)_{\Gamma_g}$ is in the kernel of ${\rm tr}_{c, \mu}$.
\end{remark}

We are also interested in the case of $\Gamma=\G$ which has torsion elements. The computation of the $K_\bullet(C(\Omega)\rtimes_\sigma\G)$ becomes more involved, and we leave this for a future work.


\end{document}